\documentclass[12pt]{article}

\usepackage{amssymb,amsmath,amsthm}
\usepackage{graphicx}
\usepackage{geometry}
\usepackage{natbib}
\usepackage{hyperref}
\usepackage{authblk}

\newtheorem{corollary}{Corollary}
\newtheorem{definition}{Definition}

\newtheorem{proposition}{Proposition}

\newcommand{\Exp}[1]{{\rm{E}}[ #1 ]}
\newcommand{\Var}[1]{{\rm{Var}}[ #1 ]}
\newcommand{\Cov}[1]{{\rm{Cov}}[ #1 ]}
\newcommand{\tr}{\text{\rm trace}}
\renewcommand{\vec}{\rm{vec}}

\title{Core Shrinkage Covariance Estimation 
for Matrix-variate Data}
\author[1]{Peter Hoff} 
\author[1]{Andrew McCormack}  
\affil[1]{Department of Statistical Science, Duke University} 
\author[2]{Anru R. Zhang} 
\affil[2]{Department of Biostatistics and Bioinformatics, Duke University} 

\begin{document}
\maketitle

\begin{abstract}
A separable covariance model for a 
random matrix
provides a parsimonious description of the 
covariances among the rows and among the columns of the matrix,  
and permits 
likelihood-based inference 
with a very small sample size. 
However, in many applications the assumption of exact separability 
is unlikely to be met, and data analysis with a separable model
may overlook or misrepresent important dependence patterns in the data.  
In this article, we propose a compromise between separable  
and unstructured covariance estimation. We show how the set of 
covariance matrices may be uniquely parametrized in terms of 
the set of separable covariance matrices and 
a complementary set of ``core'' covariance matrices, where the core of a 
separable covariance matrix is the identity matrix. 
This parametrization defines a Kronecker-core decomposition of a 
covariance matrix.
By shrinking the 
core of the sample covariance matrix with an empirical Bayes procedure, 
we obtain an estimator 
that can adapt to the degree of separability of the population 
covariance matrix.


\end{abstract}

\section{Introduction}
Many modern datasets include matrix-variate data, 
that is, a sample of $n$ matrices $Y_1,\ldots, Y_n$ 
having a common dimension $p_1\times p_2$. 
Examples of such datasets include collections of images, networks, 
gene by tissue expression arrays, and multivariate 
time series,
among others.  
One approach to the analysis of such data is to first vectorize 
each data matrix and then proceed with a method that is appropriate 
for generic multivariate data. For example, if $Y_1,\ldots,Y_n$ 
is a random sample from a population of mean-zero matrices, 
the population covariance 
could be estimated by the sample covariance 
$S=\sum_{i=1}^n y_i y_i^\top/n$, 
where for $i=1,\ldots,n$, $y_i$ is the vector of length $p=p_1\times p_2$ 
obtained by vectorizing $Y_i$.  

However, in many applications the sample size $n$ is insufficient
for such unstructured estimates to be statistically stable. 
For example, even though $p_1$ and $p_2$ might be of moderate 
magnitude individually, 
a sample size of $n\geq p_1 p_2$ is necessary for $S$ 
to be non-singular, and for the likelihood corresponding 
to a normal model to be bounded. 
Furthermore, even if the sample size is sufficient for estimation, 
an unstructured estimate such as $S$ 
may be difficult to interpret, as it is not expressed in terms of 
conceptually simple row factors or column factors. 

For these reasons, 
covariance models that are based on the matrix structure of the data
have been developed. Most popular are the separable or Kronecker-structured
covariance models 
that assume the $p \times p$ population covariance matrix is the Kronecker 
product of two smaller covariance matrices of dimension  
$p_1\times p_1$ and $p_2\times p_2$, representing across-row
and across-column covariance respectively.  
In particular, the separable covariance model 
for normally-distributed data 
\citep{dawid_1981} has been used for a wide variety of applications   
including 
environmental monitoring \citep{mardia_goodall_1993},
signal processing \citep{werner_jansson_stoica_2008}, 
image analysis \citep{zhang_schneider_2010}, 
gene expression data \citep{yin_li_2012}, 
radar detection \citep{greenewald_zelnio_hero_2016} and many others.  

In addition to its interpretability, 
a separable covariance model is appealing because of its 
statistical stability, which is a result of its parsimony as compared to 
an unstructured covariance model. 
Remarkably, the MLE in the separable normal 
model exists uniquely for any sample size $n$ larger than
$p_1/p_2 + p_2/p_1$ 
\citep{ros_bijma_demunck_degunst_2016,
soloveychik_trushin_2016,drton_kuriki_hoff_2021,
derksen_makam_2021}. 
This is in contrast to a sample size 
requirement of $n\geq p_1 p_2$ in a normal model 
with an unstructured covariance. 
However, the appropriateness of a separable covariance estimator depends 
on the extent to which the population covariance is truly separable. 
If the population covariance is not separable, a separable estimate 
might give an incomplete or misleading summary of the 
statistical dependencies in the data, or could lead to poor performance 
of statistical procedures, such as generalized least-squares or quadratic discriminant analysis, that rely on an accurate estimate of the 
population covariance. 
These and other concerns 
about the appropriateness of the separability assumption 
have been raised by \citet{stein_2005} and \cite{rougier_2017}, 
specifically in the context of random spatio-temporal processes. 
To address these concerns, 
\citet{masak_sarkar_panaretos_2022} and
\citet{masak_panaretos_2022} have proposed generalizations of the 
class of separable 
covariance operators for functional data analysis with two-dimensional 
domains (e.g., space and time).  The first of these is based on 
an approximation of 
an arbitrary positive definite covariance operator by a sum 
of separable matrices.
The second of these 
assumes the 
covariance operator is
the sum of two positive definite operators, one of which is separable 
and the other is banded, where the banding is determined by the metrics of each of the two domains. 

In this article we consider covariance estimation 
for random matrices with rows and columns that 
represent
arbitrary factors, and so 
in particular do 
not necessarily correspond to points in a spatio-temporal domain. 
We develop a covariance estimation strategy that makes 
use of the parsimony and interpretability of a separable covariance model, 
yet can consistently describe covariance matrices that are non-separable.  
This is accomplished with a new matrix decomposition for positive definite 
matrices, which we call the ``Kronecker-core decomposition'', or KCD. 
This decomposition expresses an arbitrary covariance matrix 
in terms of a low-dimensional separable covariance matrix and a complementary high-dimensional 
``core'' covariance matrix. By adaptively shrinking the 
core of the sample covariance matrix, an estimator is obtained 
that can
have a risk that is comparable to that of 
the separable estimator when the population covariance 
is truly separable, and otherwise 
has lower risk 
than 
both the separable and unstructured estimators. 

In the next section we define the Kronecker covariance 
and core covariance of an arbitrary $p_1p_2\times p_1p_2$ covariance matrix.
We show that the space of all $p_1p_2\times p_1p_2$ covariance matrices 
can be identifiably parametrized  by the product space 
of Kronecker and core covariance matrices using the Kronecker-core 
decomposition. In Section 3, we propose a class of core shrinkage estimators 
that are obtained by shrinking the core of the sample covariance matrix towards 
the identity matrix, or equivalently, shrinking the sample covariance matrix 
towards the space of separable covariance matrices. 
Such shrinkage estimators can be viewed as empirical 
Bayes estimators, where the 
amount of shrinkage is estimated from the data. 
We show that our proposed core shrinkage estimator is consistent, and 
in a simulation study in Section 4.1, 
we show that the loss of the core shrinkage estimator can be very close 
to that of an oracle Bayes estimator, and lower than that of 
both the separable and 
unstructured MLEs across a variety of conditions.  
In Section 4.2, we use core shrinkage estimators as inputs into a 
quadratic discriminant analysis for a
speech recognition task.
We observe that classifications using core shrinkage estimators have 
lower 
out-of-sample misclassification rates than those using separable or unstructured MLEs. 
A discussion of directions for further research follows in Section 5. 
Proofs of mathematical results are provided in an appendix. Replication 
code for the numerical results in this article are available at the 
first author's 
website and from the {\sf R}-package {\tt covKCD}. 

\section{Kronecker and core covariances}

\subsection{The Kronecker covariance of a random matrix}  

Let $Y$ be a mean-zero random matrix taking values in 
$\mathbb R^{p_1\times p_2}$ with a non-singular covariance matrix 
$\Sigma \in \mathcal S^+_{p}$ where $p=p_1p_2$, meaning that
$\Exp{yy^\top} = \Sigma$ where $y = \vec(Y)$.  
In what follows, we will use both 
$\Var{Y}$ and $\Var{y}$ to denote the $p \times p$ covariance matrix 
of the vectorization $y$ of $Y$.  
Recall that $\Sigma$ is \emph{Kronecker separable}, or simply 
\emph{separable}, if 
it can be expressed as $\Sigma = \Sigma_2 \otimes\Sigma_1$ for
some matrices $\Sigma_1\in \mathcal S^+_{p_1}$, 
$\Sigma_2\in \mathcal S^+_{p_2}$, where ``$\otimes$'' is the 
Kronecker product. 
In this case, the matrices $\Sigma_1,\Sigma_2$ 
(or matrices $c \Sigma_1,  \Sigma_2/c$ for any $c>0$)
are often referred to as the row covariance and column covariance
of $Y$ 
respectively. For example,  
the covariance of the  $p_2$ random variables in a common row of $Y$ 
is proportional to $\Sigma_2$, and so
$\Sigma_2$ represents the covariances of the elements 
of $Y$ across its columns. 

Let $\mathcal S_{p_1,p_2}^+ = \{ 
   \Sigma_2 \otimes \Sigma_1 : \Sigma_1\in \mathcal S_{p_1}^+, 
\Sigma_2\in \mathcal S_{p_2}^+ \} \subset \mathcal S_p^+$ be
the set of separable covariance matrices for given values of $p_1$ and $p_2$. 
A separable covariance model 
is a collection of probability distributions for $Y$ for which it is 
assumed that
$\Var{Y}  \in \mathcal S_{p_1,p_2}^+$. The most widely used 
separable model is the separable normal model, or ``matrix normal'' model \citep{dawid_1981}, 
which specifies that $Y\sim N_{p_1\times p_2}( 0 , \Sigma_2 \otimes \Sigma_1)$
for unknown $\Sigma_2\otimes \Sigma_1\in \mathcal S_{p_1,p_2}^+$.
A separable covariance model can be thought of as a bilinear 
transformation model: Let $Z$ be a $p_1\times p_2$ mean-zero 
random matrix with $\Var{Z} = I_{p}$, and let $Y = A_1 Z A_2^\top$ 
for non-singular matrices $A_1\in \mathbb R^{p_1\times p_1}$, 
$A_1\in \mathbb R^{p_2\times p_2}$. Then $\Var{Y} = A_2 A_2^\top \otimes 
 A_1 A_1^\top$, and the range of $\Var{Y}$ over all such matrices 
$A_1,A_2$ is exactly equal to $\mathcal S_{p_1,p_2}^+$. More generally, 
separability is preserved under row and column transformations of $Y$:
If $\Var{Y} = \Sigma_2\otimes \Sigma_1$, then 
\begin{align} 
\Var{ A_1 Y A_2^\top } \equiv \Var{ (A_2 \otimes A_1 ) y }  &=  
  (A_2\otimes A_1) \Var{y} (A_2\otimes A_1)^\top \nonumber  \\ 
  &=  (A_2\otimes A_1)  (\Sigma_2 \otimes \Sigma_1) (A_2\otimes A_1)^\top \nonumber   \\
 & = (A_2 \Sigma_2 A_2^\top ) \otimes (A_1 \Sigma_1 A_1^\top ).  
\label{eqn:septrans} 
\end{align} 
In the language of group theory, 
let $GL_{p_1,p_2} = \{ A_2 \otimes A_1 : 
  A_2 \in GL_{p_1},A_2\in GL_{p_2} \}$ be the 
separable subgroup of the general linear group $GL_{p}$ of 
nonsingular $p\times p$ matrices.  
The transformation in (\ref{eqn:septrans}) from 
 $\Var{Y}$ to $\Var{ A_1 Y A_2^\top }$ 
defines a transitive group action 
of $GL_{p_1,p_2}$ on $\mathcal S_{p_1,p_2}^+$. 
The group structure of the separable normal model and related 
tensor normal models has been exploited to develop methods 
for statistical estimation \citep{gerard_hoff_2015} and testing 
\citep{gerard_hoff_2016, hoff_2016a}. 

Even if $\Var{Y}$ is not separable, it still may be of interest to 
define some notion of row covariance and column covariance for 
$Y$. To this end, we identify a 
separable covariance matrix 
$K \in \mathcal S_{p_1,p_2}^+$ that summarizes 
the row and column covariance of $Y$ 
when $\Var{Y}$ is an arbitrary 
covariance matrix $\Sigma\in \mathcal S_p^+$:
\begin{definition}  
Let $\Exp{Y}=0$ and $\Var{Y} = \Sigma \in \mathcal S_p^+$. 
The Kronecker covariance of $\Sigma$ is 
$k(\Sigma) = \Sigma_2\otimes \Sigma_1$, where 
 $(\Sigma_1,\Sigma_2)$ are any matrices in 
$\mathcal S^+_{p_1}\times \mathcal S^+_{p_2}$ that 
satisfy
\begin{align} 
\label{eqn:kroneqn} 
\Sigma_1 & = \Exp{ Y \Sigma_2^{-1} Y^\top }/p_2 \\
\Sigma_2 & = \Exp{ Y^\top \Sigma_1^{-1} Y }/p_1.  \nonumber
\end{align} 
\end{definition} 
Matrices $\Sigma_1$ and $\Sigma_2$ that solve (\ref{eqn:kroneqn}) 
are weighted averages of  
across-row and across-column covariance matrices of 
whitened versions of $Y$.
For example, $\Sigma_1$ is obtained from $Y$ by first whitening  
across its columns by $\Sigma_2$.

Solutions to (\ref{eqn:kroneqn}) exist for all $\Sigma\in \mathcal S_p^+$, and all solutions 
have the same Kronecker product, and 
so the 
Kronecker covariance function 
$k: \mathcal S_{p}^+ \rightarrow  \mathcal S_{p_1,p_2}^+$ is 
well defined. The existence of solutions and uniqueness 
of their Kronecker product follow from 
existing results for the separable normal model, and the following 
alternative definition of $k(\Sigma)$ as
the element of 
$\mathcal S_{p_1,p_2}^+$ that is closest to $\Sigma$ in terms of 
a standard divergence function:
\begin{proposition} 
\label{prop:krondivergence} 
$(\Sigma_1,\Sigma_2)$ is a solution to (\ref{eqn:kroneqn}) 
if and only if $\Sigma_2\otimes \Sigma_1$ minimizes
$d(K:\Sigma) = \ln |K| + \tr(K^{-1}\Sigma)$ over $K\in \mathcal S_{p_1,p_2}^+$. 
\end{proposition}
The divergence function $d(K:\Sigma)$ is related 
to Stein's loss for covariance estimation and to the 
Kullback-Leibler divergence between two normal distributions. 
Specifically, $k(\Sigma)$ is the covariance matrix of the separable normal 
distribution that minimizes the 
Kullback-Leibler divergence to the $N_{p_1\times p_2}(0,\Sigma)$ distribution. 
This means, for example, that 
if $Y_1,\ldots, Y_n \sim$ i.i.d.\ $N_{p_1 \times p_2}(0,\Sigma)$ then 
the maximum likelihood estimator (MLE) of $\Sigma_2\otimes \Sigma_1$ 
under the  potentially misspecified model 
$Y_1,\ldots, Y_n \sim$ i.i.d.\ $N_{p_1 \times p_2}(0,\Sigma_2\otimes \Sigma_1)$ 
converges in probability to $k(\Sigma)$ as 
$n\rightarrow \infty$ \citep{huber_1967}. In the language of misspecified 
models, $k(\Sigma)$ is the ``pseudo-true'' parameter under the separable normal 
model in the case that $\Sigma$ is not necessarily separable. 

That the minimizer 
of the divergence function is unique follows from 
uniqueness results for the MLE in the 
separable normal model. The 
MLE for this model is obtained by minimizing 
over $\Sigma_2\otimes \Sigma_1\in \mathcal S_{p_1,p_2}^+$ the 
scaled log-likelihood  
\[ 
(-2/n)\times  \ln p(Y_1,\ldots, Y_n| \Sigma_2\otimes \Sigma_1)  =  
   \ln |\Sigma_2 \otimes \Sigma_1 | + \tr( (\Sigma_2\otimes \Sigma_1)^{-1} 
  S )   +   p \ln 2\pi  
\]
where $S= \sum_{i=1}^n y_iy_i^\top/n$ is the sample covariance matrix. 
Clearly, the conditions on $S$ for there to exist a unique 
MLE of $\Sigma_2 \otimes \Sigma_1$ are the same as those on $\Sigma$ 
for there 
to exist a unique minimizer of $d(K:\Sigma)$
over $K\in \mathcal S_{p_1,p_2}^+$. In particular, 
$k(S)$ is the MLE of $\Sigma_2\otimes \Sigma_1$ under the separable normal 
model when $S$ is the sample covariance matrix. 
\citet{srivastiva_vonrosen_vonrosen_2008} show that 
this MLE 
exists uniquely if $S$ is strictly positive definite, 
which implies that  $k(\Sigma)$ exists uniquely for any 
$\Sigma\in \mathcal S_p^+$. We note that solutions may also exist 
uniquely when $S$, or analogously $\Sigma$, is singular 
 \citep{
soloveychik_trushin_2016,
drton_kuriki_hoff_2021,
derksen_makam_2021}.

Numerical methods for finding 
the separable normal MLE may be used to compute the 
Kronecker covariance function. 
As shown in 
\citet{dutilleul_1999}, $\hat \Sigma_2\otimes \hat \Sigma_1$ is 
an MLE of $\Sigma_2\otimes \Sigma_1$ if 
$(\hat\Sigma_1,\hat\Sigma_2)$ satisfy 
\begin{align}
 \left (\sum_{i=1}^n Y_i \hat\Sigma_2^{-1} Y_i^\top/n\right )/p_2  & = \hat \Sigma_1  \label{eqn:kmle}  \\
 \left ( \sum_{i=1}^n Y_i^\top \hat\Sigma_1^{-1} Y_i/n \right )/  p_1 & = \hat \Sigma_2 .  \nonumber
\end{align} 
Dutilleul also provided a block coordinate descent algorithm that 
converges to the MLE when it exists uniquely. 
Because this system of equations is analogous to the 
system (\ref{eqn:kroneqn})
that define $k(\Sigma)$, 
 Dutilleul's algorithm may be implemented to numerically compute
the Kronecker covariance $k(\Sigma)$ of any $\Sigma\in \mathcal S^+_{p}$.
In this context,
given a starting value $\Sigma_2 \in \mathcal S^+_{p_2}$, 
the algorithm is to iterate the following steps until a convergence criteria is met:
\begin{enumerate}
\item  Set $\Sigma_1 =  \Exp{ Y \Sigma_2^{-1} Y^\top }/p_2;$
\item  Set $\Sigma_2 =  \Exp{ Y^\top \Sigma_1^{-1} Y }/p_1.$
\end{enumerate}
An algorithm to compute $k(\Sigma)$ is provided in the 
replication material for this article.

An important property of the Kronecker covariance function is
how it is affected by transformations of $\Sigma$, or equivalently, of $Y$. 
Recall that if $Y$ has a separable 
covariance  $\Sigma_2\otimes \Sigma_1$, then $A_1 Y A_2^\top$
has separable covariance 
$(A_2 \Sigma_2 A_2^\top) \otimes (A_1 \Sigma_1 A_1^\top)$, and so 
in this sense a linear transformation 
across the rows of $Y$ changes the 
row covariance and not the column covariance, and analogously for 
a column transformation. 
The following result shows that 
the Kronecker covariance function transforms 
in the same way, even if the covariance matrix of $Y$ is not 
separable:
\begin{proposition} 
\label{prop:kronaction} 
For 
$A_2\otimes A_1 \in GL_{p_1,p_2}$ and 
$\Sigma\in \mathcal S_{p}^+$ with $k(\Sigma) = \Sigma_2\otimes \Sigma_1$, 
\begin{align*} k( (A_2\otimes A_1) \Sigma   (A_2\otimes A_1)^\top  ) & = 
             (A_2\otimes A_1)  k(\Sigma) (A_2\otimes A_1)^\top.   \\ 
      &= (A_2 \Sigma_2 A_2^\top) \otimes (A_1 \Sigma_1 A_1^\top).
\end{align*} 
\end{proposition} 
From the perspective of group theory, the group action of 
$GL_{p_1,p_2}$ on $\mathbb R^{p_1\times p_2}$ 
defined  by $Y \mapsto A_1 Y A_2^\top$ induces a 
group action of $GL_{p_1,p_2}$ on $\mathcal S_{p}^+$ 
given by $\Sigma \mapsto (A_2\otimes A_1) \Sigma (A_2\otimes A_1)^\top$. 
The result is that the Kronecker covariance function 
$k$ is equivariant with respect to this group action -
the Kronecker covariance of the separably-transformed $\Sigma$ is 
the separably-transformed Kronecker covariance of $\Sigma$. 
 This property will 
be used throughout the remainder of this article.
Additional properties of the Kronecker covariance function include the 
following:
\begin{corollary} \
\label{cor:kroncor}
\begin{enumerate} 
\item $k(I_p)=I_p$. 
\item If $\Sigma \in \mathcal S_{p_1,p_2}^+$ then 
$k(\Sigma) = \Sigma$. 
\item For $a>0$, $k(a \Sigma ) = a k(\Sigma)$. 
\item If $\Sigma$ is diagonal then $k(\Sigma)$ is diagonal.
\end{enumerate}
\end{corollary}
The third item indicates that $k$ is a scale-equivariant function. 
As a result, 
the shrinkage estimator we propose in the Section 3
will be scale-equivariant. 

\subsection{The Kronecker-core parametrization and decomposition}
The Kronecker covariance function $k$ defined above is 
a surjection from $\mathcal S^+_p$ to $\mathcal S^+_{p_1 \times p_2}$
that describes the row covariance and column covariance 
of an arbitrary element of $\mathcal S_p^+$. 
We now use this function to define, for each 
$\Sigma\in \mathcal S^+_p$, 
a ``core'' covariance matrix $c(\Sigma)$ 
that is complementary to $k(\Sigma)$ in that the core lacks 
across-row and across-column covariance in some sense. 
We then show that, taken together, 
the product space of separable and core covariance matrices identifiably parametrizes $\mathcal S_p^+$.

A core covariance of $\Sigma\in \mathcal S_p^+$ 
is obtained by applying a transformation to $\Sigma$ that whitens its
Kronecker covariance. 
Specifically, 
let $H = H_2\otimes H_1$ 
be any matrix in $GL_{p_1,p_2}$ such that 
$HH^\top = k(\Sigma)$.  
By the equivariance of 
$k$, we have
\begin{align*}  
k(H^{-1} \Sigma H^{-\top} ) & =  H^{-1} k(\Sigma) H^{-\top} \\
& =  H^{-1} H H^\top H^{-\top} = I_p. 
\end{align*}
We define Kronecker-whitened versions of 
$\Sigma$ as follows:
\begin{definition} 
Let $H=H_2\otimes H_1\in GL_{p_1,p_2}$ satisfy $HH^\top = k(\Sigma)$. Then the matrix $C$ given by 
$C =H^{-1} \Sigma H^{-\top}$  is a \emph{core} of $\Sigma$. 
\end{definition} 
We call the matrix $C=H^{-1} \Sigma H^{-\top}$ 
a core of $\Sigma$ because
the 
four-way tensor $\tilde \Sigma\in \mathbb R^{p_1\times p_2 \times p_1\times p_2}$ with entries corresponding to $\Sigma$ may be expressed in 
terms of $C$, $H_1$ and $H_2$ 
via the multilinear operation   
\[
 \tilde \Sigma = \tilde C \times \{ H_1,H_2, H_1, H_2 \}, 
\]
where ``$\times$'' is the multilinear product 
and $\tilde C$ is the four-way tensor corresponding to $C$. 
Equivalently, we have 
$\vec(\Sigma) = (H_2 \otimes H_1 \otimes H_2 \otimes H_1 ) \vec(C)$. 
In the context of Tucker products, the tensor that gets multiplied 
along each mode by a matrix is called the ``core''.

There are multiple cores for a given $\Sigma$, as there are multiple 
separable matrices
$H$ for which $HH^\top= k(\Sigma)$. Conversely,  
a core of $\Sigma$ is also a core of $H\Sigma H^\top$ for 
any $H\in GL_{p_1,p_2}$. 
More generally, we say that a covariance matrix $C\in \mathcal S^+_p$ 
is a core covariance 
matrix if it is the core of some $\Sigma\in \mathcal S^+_p$, 
and so any core covariance matrix satisfies 
$k(C)=I_p$. Furthermore, suppose $C\in \mathcal S^+_p$ satisfies 
$k(C)=I_p$. Then $C$ is a core of any covariance matrix $HCH^\top$ 
for $H \in \mathcal S^+_{p_1,p_2}$. 
As such, for a given $p_1$ and $p_2$, we define the set of 
core matrices as follows:
\begin{definition} 
For a given $p_1$ and $p_2$ with $p_1\times p_2=p$, the set of 
core covariance matrices is $\mathcal C^+_{p_1,p_2} = \{ 
C \in \mathcal S_{p}^+ : k(C) = I_p \}$. 
\end{definition}  
The condition $k(C)=I_p$ 
defining $\mathcal C_{p_1,p_2}^+$ 
can alternatively be expressed as follows:
\begin{proposition}
\label{prop:coreid} 
Let $Y$ have 
covariance matrix $C \in \mathcal S_{p}^+$, 
and let $\tilde C$ be the $p_1\times p_2 \times p_1 \times p_2$ 
tensor where $\tilde C_{i,j,i',j'}= \Cov{ Y_{i,j},Y_{i',j'} }.$
Then $k(C) = I_p$ if and only if 
\begin{align*}
\Exp{YY^\top}/p_2 & \equiv 
 \sum_{j=1}^{p_2} \tilde C_{,j,,j}/p_2 = I_{p_1}  \\
\Exp{Y^\top Y}/p_1 & \equiv  \sum_{i=1}^{p_1} \tilde C_{i,,i,}/p_1 = I_{p_2}. 
\end{align*}
\end{proposition} 
So for a core covariance matrix, 
the across-column average of the across-row covariance
is the identity matrix, and analogously 
for the across-column covariance. Intuitively, 
a core covariance has no across-row or across-column correlation 
or heteroscedasticity, on average. 

From the proposition 
we see that $\mathcal C_{p_1,p_2}^+$ is 
defined by a system of linear constraints, 
and that $\tr(C) = p$ for 
any core covariance matrix $C$, so $\mathcal C_{p_1,p_2}^+$
is a compact convex subset of $\mathcal S_{p}^+$. 
Additionally, 
the core covariances  $\mathcal C^+_{p_1,p_2}$ and 
the separable covariances $\mathcal S^+_{p_1,p_2}$ 
are nearly non-overlapping:
If $C$ is core then $k(C)=I_p$, and if $C$ is separable, 
then $k(C) =C$ by Corollary \ref{cor:kroncor}. Therefore, if $C$ is core and separable, then 
$C=I_p$. Thus $\mathcal S^+_{p_1,p_2} \cap \mathcal C^+_{p_1,p_2} = I_p$.

For every $\Sigma\in \mathcal S_p^+$ there 
is a core matrix $C\in \mathcal C^+_{p_1,p_2}$ and 
separable matrix 
$K \in \mathcal S^+_{p_1,p_2}$ such that
$\Sigma =  H C H^\top$ for some separable 
$H\in GL_{p_1,p_2}$ such that $K=H H^\top$. 
Conversely, to every $C\in \mathcal C^+_{p_1,p_2}$ and $K\in
\mathcal S^+_{p_1,p_2}$ we can define an element of $\mathcal S^+_p$ as 
$HCH^\top$ where $H\in GL_{p_1,p_2}$ and $K=HH^\top$. This suggests that there is a bijection 
between $\mathcal S^+_p$ and 
$ \mathcal S^+_{p_1,p_2}\times  \mathcal C^+_{p_1,p_2}$. 
In fact, there are many such bijections, including one for each 
way to define a separable matrix square root $H$ of $K$, 
or equivalently, one for each way to define a row and column whitening matrix 
from $K$.  
To specify a particular bijection, we need to specify a separable 
square root function.
\begin{definition}  
Let $\mathcal H$ be a subset of $GL_{p_1,p_2}$ such that
the function 
$s:\mathcal H\rightarrow \mathcal S^+_{p_1,p_2}$ defined by 
$s(H) = HH^\top$ is a bijection. 
Then $h = s^{-1}$ is a separable matrix square root function. 
\end{definition}
Essentially, a separable square root function is defined by 
a set of separable 
matrices $\mathcal H$ with unique crossproducts, the set of which 
equals the set of separable covariance matrices. 
The defining feature of such a function is that $h(HH^\top) = H$ 
for $H\in \mathcal H$. 
Examples include the following:
\begin{itemize}
\item Symmetric square root: $h(\Sigma_2\otimes \Sigma_1)=\Sigma_2^{1/2} \otimes \Sigma_1^{1/2}$, where $\Sigma_j^{1/2}$ is the symmetric square 
root of $\Sigma_j$, $j\in\{1,2\}$. 
  \item Cholesky square root: $h(\Sigma_2\otimes \Sigma_1)= 
   L_2\otimes L_1$  
where $L_j L_j^\top$ is the 
lower triangular Cholesky factorization of $\Sigma_j$, $j\in \{1,2\}$. 
 \item PCA square root: $h(\Sigma_2\otimes \Sigma_1)= 
 E_2 \Lambda_2^{1/2} \otimes E_1 \Lambda_1^{1/2}$ 
where $E_j \Lambda_j E_j^\top$ is the 
eigendecomposition of $\Sigma_j$, $j\in \{1,2\}$. 
Note that conventions on the signs and column orderings of $E$ 
need to be specified in order for $h$
to be a bijection. 
\end{itemize}      
A non-example would be $GL_{p_1,p_2}$: While the set of crossproducts of 
this set is equal to $\mathcal S^+_{p_1,p_2}$, elements of the set do 
not have unique crossproducts. 

For a given separable square root function $h$ we define the core 
covariance function $c:\mathcal S_p^+\rightarrow \mathcal C_{p_1,p_2}^+$ 
as $c(\Sigma) = H^{-1} \Sigma H^{-\top}$ where $H = h(k(\Sigma))$. 
Since the core represents ``non-separable'' covariance, 
we would hope the core function to be invariant 
to bilinear transformations 
of the form $\Sigma \mapsto (A_2\otimes A_1) \Sigma (A_2\otimes A_1)^\top$
that induce separable covariance. 
This property partly holds:
\begin{proposition}
\label{prop:coreprop}  
Let $\Sigma\in \mathcal S_{p}^+$ and 
$A\in GL_{p_1,p_2}$. Then 
\begin{enumerate}
\item $c(A \Sigma A^\top ) = (R_2\otimes R_1) c(\Sigma) (R_2\otimes R_1)^\top$
 for some 
$R_1\in \mathcal O_{p_1}$, $R_2\in \mathcal O_{p_2}$. 
\item $c(A \Sigma A^\top ) = c(\Sigma)$ if $A\in \mathcal H$ and 
$\mathcal H$ is a group. 
\end{enumerate}
\end{proposition}  
Item 2 of the proposition says that if  $\mathcal H$ is a group then 
$c$ is a maximal invariant function of $\Sigma\in \mathcal S_p^+$ under the 
group action $\Sigma \mapsto H \Sigma H^\top$ for $H\in \mathcal H$, 
while the Kronecker covariance function $k$ is an equivariant function 
by Proposition \ref{prop:kronaction}. One such group $\mathcal H$ 
is the set of Kronecker products of lower-triangular matrices with 
positive diagonal entries, with $h$ being the Cholesky square root. 
However, 
the results on covariance estimation in the 
remainder of the article are unaffected by the 
choice of $h$ as long as it is
continuous, as is the case for the 
symmetric and Cholesky square root functions 
mentioned above. 
We assume use of one of these two
 continuous square root function for the remainder of the article. 

We now arrive at the main result of this section - an identifiable 
parametrization of the set of covariance matrices in terms of 
Kronecker and core covariance matrices:
\begin{proposition} 
\label{prop:reparam} 
The function 
$f:\mathcal  S^{+}_p 
\rightarrow \mathcal S^+_{p_1,p_2} \times \mathcal C^+_{p_1, p_2}$ 
defined by $f(\Sigma) =  ( k(\Sigma),c(\Sigma))$ is a 
homeomorphism 
with inverse 
$ g:\mathcal S_{p_1,p_2}^+ \times \mathcal C_{p_1,p_2}^+ 
 \rightarrow \mathcal S_p^+$  given by 
$g(K,C) = h(K) C h(K)^\top$. 
\end{proposition}
The function $g$ can be viewed as a 
parametrization of $\mathcal S_p^+$ in terms 
of $\mathcal S^+_{p_1,p_2}\times\mathcal C^+_{p_1,p_2}$. 
Conversely, the function $f$ 
provides a matrix decomposition for elements of $\mathcal S_p^+$: 
Every $\Sigma \in \mathcal S_p^+$ has a unique representation 
as $\Sigma = K^{1/2} C K^{1/2}$  for some $K\in \mathcal S_{p_1,p_2}^+$
and $C\in \mathcal C_{p_1,p_2}^+$. We refer to this representation 
as the Kronecker-core decomposition, or KCD.

\section{Core shrinkage via empirical Bayes}  
\subsection{Core shrinkage estimators}
Let $Y_1,\ldots,Y_n$ be an i.i.d.\ random sample from 
a mean-zero normal population of $p_1\times p_2$ matrices, that is 
\begin{equation}
  Y_1,\ldots, Y_n \sim  \text{ i.i.d.\ } N_{p_1\times p_2}(0, \Sigma) 
\label{eqn:mvnorm}
\end{equation}
for some unknown $\Sigma \in \mathcal S^+_{p}$ where $p=p_1\times p_2$.  
In this case where no assumptions are made on the structure of 
$\Sigma$, the 
standard estimator 
is the sample covariance matrix 
$S = \sum_{i=1}^n y_i y_i^\top/n$ where 
$y_i = \text{vec}(Y_i)$. This estimator is 
unbiased, and if $n\geq p$ it is the MLE. 
However, as is well known, the risk of $S$ can be suboptimal, 
and substantially so 
if $n$ is not somewhat larger than $p$.  

As an alternative to $S$, we propose an estimator obtained by 
shrinking $S\in \mathcal S_p^+$ towards the lower-dimensional 
subset  $\mathcal S_{p_1,p_2}^+$ of separable 
covariance matrices, using the parametrization of $\mathcal S^+_p$ 
described in the previous section. 
Let $K=k(\Sigma)$ and $C=c(\Sigma)$
be the unknown Kronecker and core covariance matrices 
of $\Sigma$, 
so that $\Sigma = K^{1/2} C K^{1/2}$
where $K^{1/2}$ is the symmetric square root of $K$. 
Because MLEs are parametrization invariant, 
the MLE $(\hat K, \hat C)$ of $(K,C)$ is 
\begin{align*}
\hat K  &= k(S) \\  
\hat C  &= c(S) = \hat K^{-1/2} S \hat K^{-1/2}. 
\end{align*}
Note that $\hat K$ is also the MLE of $\Sigma$ under the 
separable normal model that assumes $\Sigma \in \mathcal S_{p_1,p_2}^+$. 
The number of parameters needed to define $K$ and to define $C$ 
are on the order of 
$p_1^2+p_2^2$ and $p_1^2 p_2^2$ respectively, and so heuristically 
we expect that $\hat K$ is a better estimate 
of $K$ than  $\hat C$ is of $C$. 
For this reason, we consider shrinkage estimators of the 
form 
\begin{align*} 
 \hat \Sigma & = \hat K^{1/2} \hat C_w \hat K^{1/2} , \\
\hat C_w & =   (1-w) \hat C + w I_p  
\end{align*}
for some choice of  $w\in [0,1]$. 
Because the space of core matrices is convex and includes $I_p$, 
the value $\hat C_w$ is itself a core matrix and 
is a linear shrinkage estimator of $C$, shrinking the  
core $\hat C$ of the sample covariance matrix 
towards $I_p$, or equivalently, 
shrinking the sample covariance matrix $\hat\Sigma$ towards the MLE $\hat K$ of the 
separable submodel: 
\begin{align} 
  \hat \Sigma  & =  \hat K^{1/2} \left [ (1-w)\hat C + w I_p \right ] 
   \hat K^{1/2}  \nonumber  \\
& = (1-w)  \hat K^{1/2} \hat C \hat K^{1/2} + w \hat K  \nonumber \\ 
&= (1-w ) S + w \hat K.   
\label{eqn:cse} 
\end{align}
In particular, $w=1$ 
gives the MLE under the assumption that $\Sigma$ is separable, whereas 
$w=0$ gives the sample covariance, or equivalently, the 
unrestricted MLE in the case that $n\geq p$. 
Furthermore, if $w>0$ then the estimator is positive 
definite even if $n$ is much smaller than $p$, 
as long as $n$ is large enough for 
$\hat K$ to be the MLE for the separable submodel. 

\subsection{Empirical Bayes estimation}
An estimator having the form (\ref{eqn:cse}) can be viewed as an 
empirical Bayes estimator. Consider an inverse-Wishart 
prior distribution for the 
unknown covariance $\Sigma$, 
\begin{equation} 
\Sigma^{-1} \sim \text{Wishart}_p( [(\nu-p-1) \Sigma_2 \otimes \Sigma_1]^{-1}, 
  \nu ), 
\label{eqn:kronprior} 
\end{equation}
which is parametrized so that 
$\Exp{ \Sigma } = \Sigma_2\otimes \Sigma_1.$
The hyperparameter $\nu$ partly controls  how concentrated the prior 
distribution of $\Sigma$ is around the separable covariance matrix
$\Sigma_2\otimes \Sigma_1$. Under this prior distribution, the posterior 
distribution of $\Sigma$ is 
\[
\Sigma^{-1}| S \sim \text{Wishart}_p([ nS+(\nu-p-1)\Sigma_2\otimes \Sigma_1
   ]^{-1}, n+\nu ), 
\]
and the Bayes estimator under squared-error loss is the posterior mean, 
\begin{equation}
\Exp{ \Sigma |S} = (1- w)S + w  \Sigma_2 \otimes \Sigma_1  
\label{eqn:obayes}
\end{equation}
where $w = (\nu-p-1)/(n+\nu-p-1 )$. 
An empirical Bayes estimator that replaces the hyperparameter 
$\Sigma_2\otimes \Sigma_1$ with $\hat K=k(S)$ gives the estimator 
\[ \widehat{ \Exp{\Sigma| S} }  =   
  (1-w) S + w \hat K,   \]
which is the same as in (\ref{eqn:cse}). 
While $\hat K$ is not the marginal MLE 
of $\Sigma_2\otimes \Sigma_1$
under the 
prior distribution (\ref{eqn:kronprior}), 
$\hat K$ can be seen 
as a marginal moment estimator in the following sense:
Because $\Exp{\Sigma} = \Sigma_2\otimes \Sigma_1$ 
the marginal variance of a generic $Y$ is $\Sigma_2\otimes \Sigma_1$ as well. 
By Corollary \ref{cor:kroncor}, we then have 
\begin{align*}
\Sigma_1& = \Exp{ Y \Sigma_2^{-1} Y^\top } /p_2 \\
\Sigma_2& = \Exp{ Y^\top \Sigma_1^{-1} Y } /p_1, 
\end{align*}
where the expectation is with respect to the marginal distribution of 
$Y$. The value $\hat K = k(S)$ is a solution to these equations 
with the marginal expectation replaced by expectation with respect to the 
empirical distribution of $Y_1,\ldots, Y_n$, and 
so $\hat K$ is a generalized method of moments estimator 
of $\Sigma_2 \otimes \Sigma_1$.

The amount of shrinkage $w$ is determined by the hyperparameter 
$\nu$. 
Our proposed empirical Bayes estimator of $\nu$ is the maximizer 
in $\nu$ of the 
marginal density $p(S|\nu,\Sigma_2\otimes \Sigma_1)$ with
$\hat K$ plugged-in for $\Sigma_2\otimes \Sigma_1$. 
This density has an 
essentially closed-form expression due to the conjugacy of the 
inverse-Wishart prior distribution (\ref{eqn:kronprior}). 
Using standard calculations, we obtain the 
marginal density of $S$ as
\[
p(S | \nu, \Sigma_2\otimes \Sigma_1) = r  \times 
 \frac{ k(\nu) |(\nu-p-1)\Sigma_2\otimes \Sigma_1|^{\nu/2} }
  { k(\nu+n) | nS +  (\nu-p-1)\Sigma_2\otimes \Sigma_1   |^{(\nu+n)/2} }.
\]
where $r$ does not depend on $\nu$ and 
$k(\nu)^{-1} =  2^{\nu p/2} \Gamma_p(\nu/2)$, 
with $\Gamma_p$ being the multivariate gamma function. 
Now we plug-in $\hat K$ for $\Sigma_2\otimes \Sigma_1$, 
and utilize the fact that $S = \hat K^{1/2} \hat C \hat K^{1/2}$ 
to obtain 
\begin{align*}
p(S|\nu,\hat K )/b & =  
|\hat K|^{-n/2} \times 
 \frac{k(\nu)}{k(\nu+n)} (\nu-p-1)^{-np/2} 
  | I_p + \tfrac{n}{\nu-p-1} \hat C |^{-(\nu+n)/2} \\
&= 
\left(     |\hat K|^{-n/2} 2^{np/2}   \right )  \times 
  \frac{ \Gamma_p((\nu+n)/2) }{\Gamma_p(\nu/2) } 
  (\nu-p-1)^{-np/2} | I_p + \tfrac{n}{\nu-p-1} \hat C |^{-(\nu+n)/2}. 
\end{align*} 
After some additional manipulation, we have that 
$p(S|\nu, \hat K  ) \propto_\nu L(\nu)$ where 
\begin{equation} 
\label{eqn:nuobj}
 L(\nu) = \frac{ \Gamma_p((n+\nu)/2)}{\Gamma_p(\nu/2)}  \times 
 w^{\nu p/2} (1-w)^{np/2}   \times | (1-w) \hat C + w I_p |^{-(\nu+n)/2}, 
\end{equation}
with $w = (\nu-p-1)/(n+\nu-p-1)$ as before. 
Computation of $L(\nu)$ is facilitated by noting 
that the determinant term 
can be expressed as 
$|  (1-w) \hat C + w I_p| = \prod_{j=1}^p  (w+ (1-w)\hat c_j)$, 
where $\hat c_1,\ldots, \hat c_p$ are the eigenvalues of $\hat C$. 
Our proposed empirical Bayes estimator of $\nu$ is the maximizer 
$\hat \nu$ of $L$, which gives $\hat w = (\hat\nu -p-1)/(n+\hat\nu - p -1)$ 
as the amount of shrinkage. 
The resulting 
empirical Bayes core shrinkage estimator is given by 
\begin{align}
\hat\Sigma 
& = \hat K^{1/2} [ (1-\hat w ) \hat C  + \hat w I_p ] \hat K^{1/2} 
\label{eqn:ebest} 
\\
& = (1-\hat w) S + \hat w \hat K.  \nonumber
\end{align} 

To understand how the data influence the value of $\hat \Sigma$ through 
 $\hat w$, 
write $L(\nu) = a(\nu)\times b(\nu)$ where 
$b(\nu) = | (1-w) \hat C + w I_p |^{-(\nu+n)/2}$ is the part of 
$L$ that depends on the data, and $a(\nu)= L(\nu)/b(\nu)$. 
The function $a(\nu)$ is generally increasing, 
and so this part of $L$ ``favors'' large
values of  $\hat \nu$ (and $\hat w)$. 
If the sample covariance $S$ is very close to being separable, then 
$\hat C$ is very close to the identity matrix and 
so $b(\nu)$ is roughly constant in $\nu$. 
In this case, $a(\nu)$ dominates $L(\nu)$, resulting 
in a large $\hat w$ and 
strong shrinkage of $S$ towards the sample Kronecker covariance 
$\hat K$. However, if $\hat C$ is far from the identity 
then $b$ can be strongly decreasing in $\nu$, which results
in $\hat \Sigma$ being close to $S$. In summary, the degree of shrinkage 
towards the space of separable covariance matrices depends on 
how close $S$ is to being separable, as measured by how close 
$\hat C$ is to the identity matrix.

Finally, we note that $\hat \Sigma$  does not 
depend on the choice of separable square root function:
This is because if $\hat C$ and $\hat C'$ are core matrices  
of $S$ obtained from different square root functions, they 
still must satisfy 
$\hat C' = R \hat C R^\top$ for some orthogonal matrix $R$. 
This difference does not affect the empirical Bayes estimator 
of $\nu$, since 
\begin{align*}
| (1-w ) \hat C' + w I_p | &=  | (1-w ) R\hat C R^\top  + w R R^\top |  \\
   &= | R R^\top|   | (1-w) \hat C + w I_p | =  | (1-w) \hat C + w I_p |. 
\end{align*}

\subsection{Consistency} 
We now provide some consistency results for 
the components of the KCD
and  the core shrinkage estimator. 
First, we have the very general result that 
a consistent estimator of $\Sigma$ can be used to obtain
consistent estimators of $k(\Sigma)$ and $c(\Sigma)$, and vice versa:
\begin{corollary}
\label{cor:khatconsistency}
$ k(S) \stackrel{p}{\rightarrow} k(\Sigma)$ and 
$ c(S) \stackrel{p}{\rightarrow} c(\Sigma)$ 
if and only if 
$S \overset{p}{\rightarrow} \Sigma$. 
\end{corollary}
This follows directly from the continuity result in Proposition 
\ref{prop:reparam} 
and the continuous mapping theorem. 
This result can be used to show the consistency of the 
core shrinkage estimator $\hat \Sigma$:
Recall that our core shrinkage estimator can be written as 
$\hat\Sigma =(1-\hat w) S + \hat w \hat K$, where 
$S$ is the sample covariance matrix and $\hat K$ is the Kronecker covariance
of $S$. 
Consistency of $\hat\Sigma$ will follow if 
$S$ is consistent and 
the weight $\hat w$ on the separable matrix $\hat K$
converges to zero if $\Sigma$ is not exactly separable. This is because 
if 
$\Sigma\not\in \mathcal S_{p_1,p_2}^+$ but $\hat w\stackrel{p}{\rightarrow} 0$
then $\hat\Sigma \stackrel{p}{\rightarrow} \Sigma$ because 
$S$ is consistent for $\Sigma$, regardless of the separability of $\Sigma$. 
Conversely, 
if $\Sigma$ is separable then $k(\Sigma)=\Sigma$ and so 
$\hat K$ is consistent  
for $\Sigma$ by the above proposition. 
To summarize, we have the following:
\begin{proposition} 
\label{prop:consistency} 
If  $S\stackrel{p}{\rightarrow} \Sigma$
then  $\hat\Sigma\stackrel{p}{\rightarrow} \Sigma$ for 
 any $\Sigma\in \mathcal S_{p}^+$. 
If  $S\stackrel{p}{\rightarrow} \Sigma$ and 
$\Sigma \not\in\mathcal S^+_{p_1,p_2}$ then 
$\hat w \stackrel{p}{\rightarrow} 0$. 
\end{proposition}
Note that the results given in Corollary \ref{cor:khatconsistency} and Proposition \ref{prop:consistency} 
only assume that $S$ converges in probability to $\Sigma$ as some index 
$n$, 
used in the definition of $\hat w$,  goes to infinity. 
The estimator $S$ need not be 
Wishart-distributed or even a sample covariance matrix, 
and $n$ need not be a sample size.

\section{Numerical examples} 

\subsection{Monte Carlo study}    
Because of its adaptive nature, we expect that 
the core shrinkage estimator $\hat\Sigma$ outperforms 
the unrestricted MLE $S$ in general, and 
performs nearly as well as the separable MLE $\hat K$ when the 
true covariance is exactly separable.  
We examine this in a finite sample setting with a small simulation 
study. We considered two dimensions for the sample space, 
$(p_1,p_2)=(5,7)$ and $(p_1,p_2) = (13,17)$ 
which correspond to values of $p=p_1\times p_2$ being 
35 and 221, respectively. For each dimension, eight 
sample sizes $n$ were considered, ranging from 
$p_2$ to $3 p_1 p_2/2$.  
For each dimension and each sample size, population 
covariance matrices were generated under four scenarios, 
three of which were simulated
from the inverse-Wishart prior distribution (\ref{eqn:kronprior}) 
with three values of the degrees of freedom parameter $\nu$
ranging from $p+2$ to $3p+1$. 
In the fourth scenario, 
which we refer to as $\nu=\infty$, 
$\Sigma$ was set to a separable matrix (the identity matrix). 
To summarize, 
our simulation scenarios include 
$8\times 4=32$  combinations of $n$ and $\nu$ for each 
of 2 different values of $(p_1,p_2)$. 

For each of these 64 
scenarios, 200 matrices $\Sigma$ were simulated from 
(\ref{eqn:kronprior}), and from each a sample of $n$ 
random matrices from the corresponding multivariate normal 
distribution (\ref{eqn:mvnorm}) were generated. 
From each sample, we computed four estimators:
the sample covariance or MLE $S$, the separable 
MLE $\hat K$,
the core shrinkage estimator $\hat \Sigma$, and 
the oracle Bayes estimator (\ref{eqn:obayes}) which uses 
perfect knowledge of the hyperparameters  
$\nu$ and $\Sigma_2\otimes \Sigma_1$ of the prior distribution 
(\ref{eqn:kronprior}). 
For each sample and estimator, the squared error loss 
in estimating $\Sigma$ was computed. 

Before comparing the estimators in terms of loss, we 
first examine the performance of the empirical Bayes estimator 
of $\hat w$ of $w$, which determines 
the amount of shrinkage towards $\hat K$. 
Results for all simulation scenarios are shown in 
Figure \ref{fig:nuest}, 
where sample means of the 200 values of $\hat w$ are plotted 
as a function of the sample size. 
On average, $\hat w$ overestimates $w$
with the bias 
decreasing with increasing sample size and dimension $p$, 
and also being 
smaller for the smaller values of $w$.
Our intuition regarding the overestimation is that
the ideal estimate of $w$ would be obtained 
by evaluating how close $S$ is to $\Sigma_2\otimes \Sigma_1$. 
In contrast, $\hat w$ is obtained by evaluating 
how close $S$ is to $\hat K$. Since $\hat K$ is 
the closest element of $\mathcal S_{p_1,p_2}^+$ to $S$ 
by construction, 
$\hat w$ overestimates how close $S$ is to 
$\Sigma_2\otimes\Sigma_1$.

\begin{figure}
\centering{\includegraphics[height=3in]{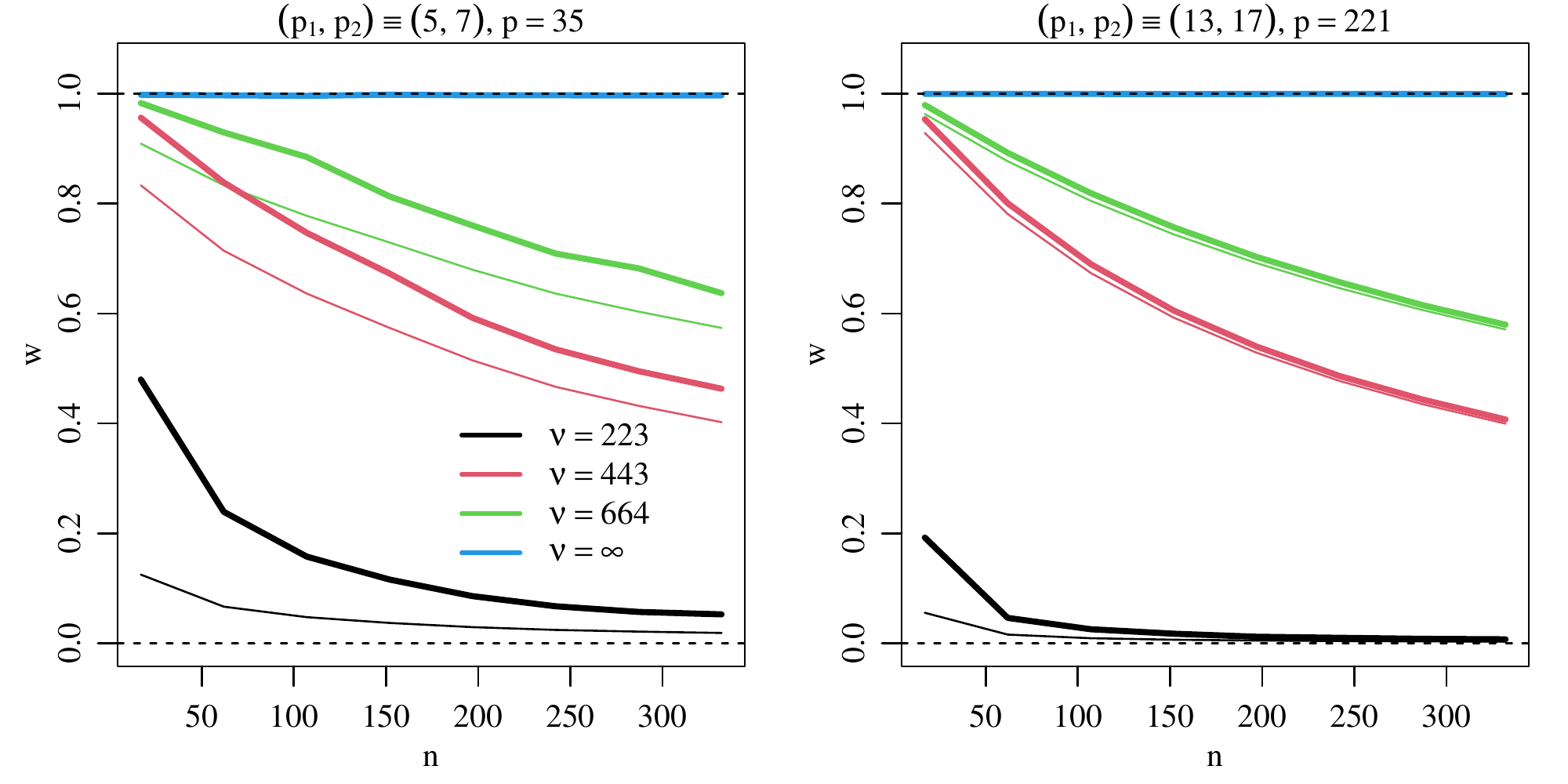}}  
\caption{Average of 200 values of $\hat w$ as a function of $\nu$ and 
$n$ in thick lines, true values of $w$ in thin lines. } 
\label{fig:nuest} 
\end{figure}

\begin{figure}
\centering{\includegraphics[height=6in]{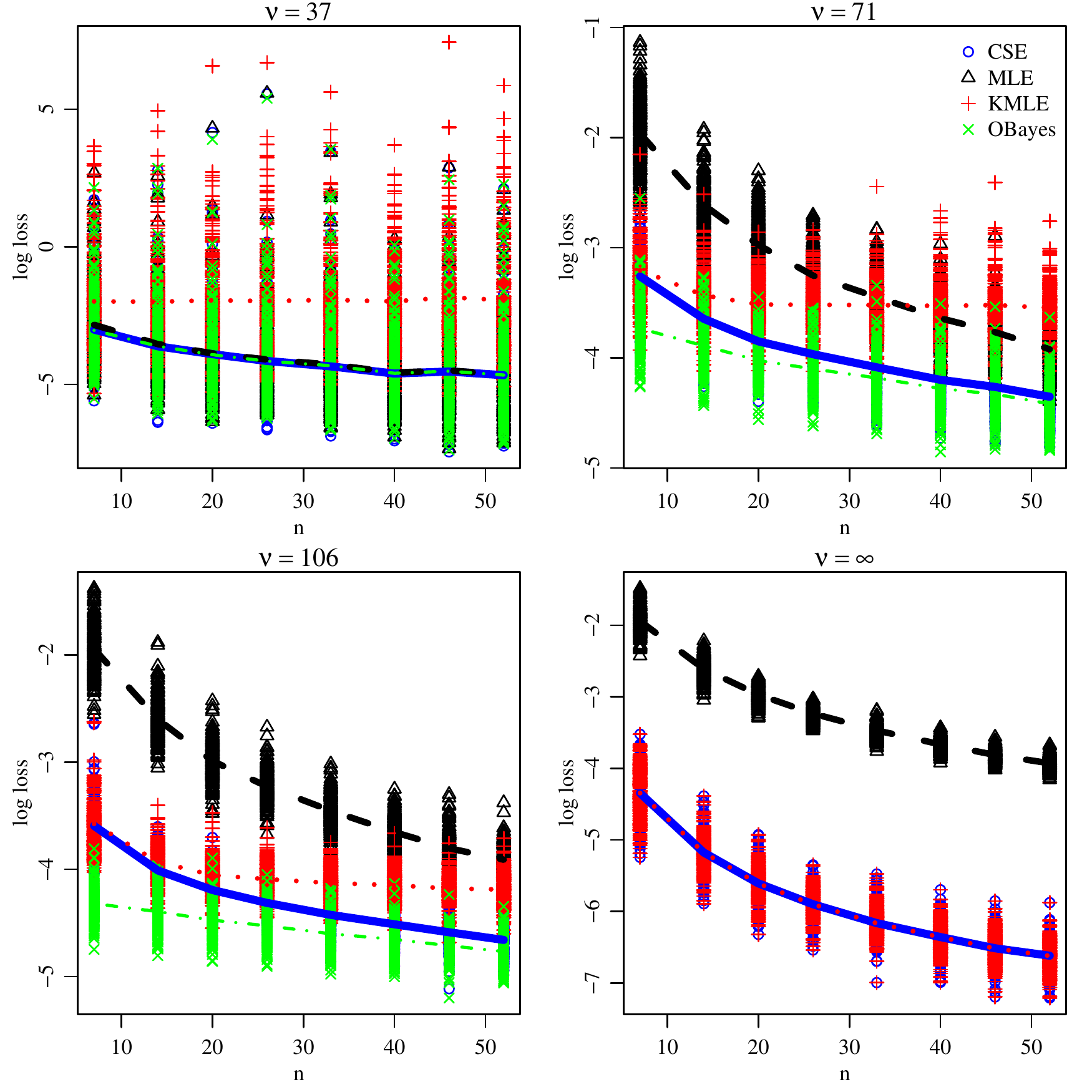}}  
\caption{Loss comparisons for $(p_1,p_2)= (5,7)$. Plotting 
symbols are given for 200 simulated datasets for each scenario. Lines 
are averages of log-loss. The estimators include the sample covariance matrix (MLE), the separable MLE (KMLE), the core shrinkage estimator (CSE) and the oracle Bayes estimator (OBayes).} 
\label{fig:lossSmall}
\end{figure} 

\begin{figure}
\centering{\includegraphics[height=6in]{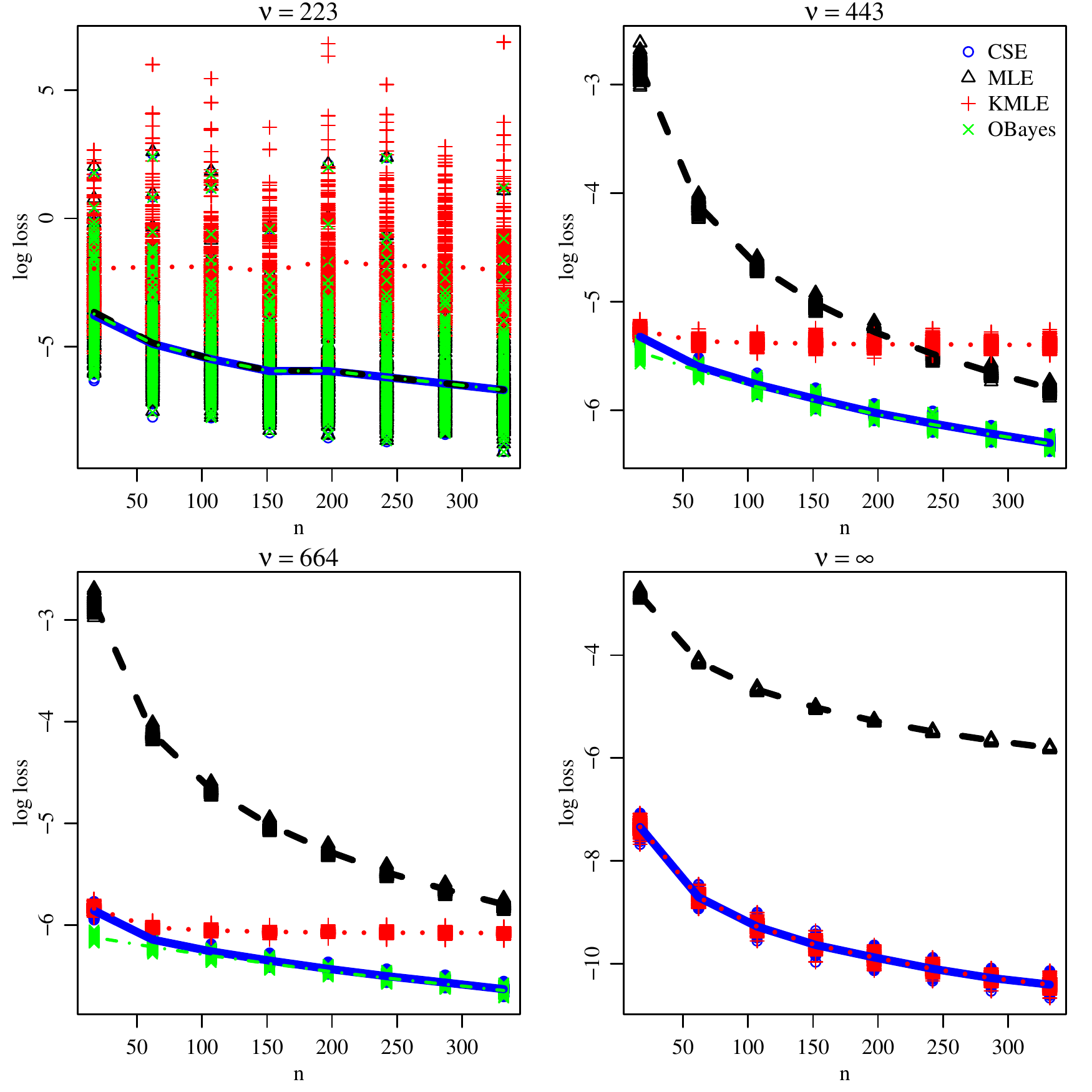}} 
\caption{Loss comparisons for $(p_1,p_2)= (13,17)$. 
Plotting 
symbols are given for 200 simulated datasets for each scenario. Lines 
are averages of log-loss. The estimators include the sample covariance matrix (MLE), the separable MLE (KMLE), the core shrinkage estimator (CSE) and the oracle Bayes estimator (OBayes).} 
\label{fig:lossMed}
\end{figure}
 
Loss comparisons for the four estimators are displayed in 
Figures \ref{fig:lossSmall} and \ref{fig:lossMed} for 
the $(p_1,p_2)=(5,7)$ and  $(p_1,p_2)=(13,17)$  scenarios 
respectively. The performance comparisons among the 
four estimators 
are similar in each of these two cases. 
The oracle Bayes estimator has the best
performance for each value of 
$\nu$.
For the smallest values of $\nu$, 
for which 
$\Sigma$ is not close to being separable, 
the performance of the unrestricted MLE is nearly identical to that of the 
oracle Bayes estimator. This is because
the value of the oracle shrinkage weight is $1/n$ 
and so these two estimators are nearly the same. 
The core shrinkage estimator (CSE) has a loss performance nearly 
identical to these two estimators, since for 
small values of $\nu$, the estimate
$\hat\nu$ is quite good. In contrast, 
the Kronecker separable MLE (KMLE) 
has worse performance on average than 
the other estimators, and its loss does not improve with 
increasing sample size. The explanation for this is that 
the Kronecker covariance $k(\Sigma)$ does not 
require a large sample size to be well-estimated, and 
so $\hat K$ is close to $k(\Sigma)$ for all sample sizes, but 
this is far from $\Sigma$ since $\Sigma$ is not close to being separable. 

The pattern changes somewhat for the larger values 
of $\nu$. In general, the loss of the KMLE
is good for small sample sizes, but does not 
improve much with increasing sample size since 
it converges to $k(\Sigma)$, which is not equal to  $\Sigma$. 
In contrast, 
the unrestricted MLE is poor for small sample sizes 
but, since it is a consistent estimator, 
has a loss that steadily decreases with increasing sample 
size. The core shrinkage estimator is generally 
as good or better than either of these estimators 
across the different sample sizes: For small $n$ 
it is about as good as the KMLE, and for 
large $n$, where both $k(\Sigma)$ and 
$\nu$ can be well-estimated, it performs 
nearly as well as the oracle Bayes estimator. 

Finally, the bottom-right panel of each figure 
gives the performance of the CSE
and unrestricted and separable MLEs in the 
case that $\Sigma$ is truly separable 
(the oracle Bayes estimator in this case is exactly 
$\Sigma$). The performance of the 
CSE and KMLE are 
nearly identical, and much better than 
that of  the unrestricted MLE. This is not too surprising
given the observation 
from Figure \ref{fig:nuest} 
that $\hat \nu$ tends to overestimate
$\nu$ when $\nu$ is a large (finite) value. 
Although any finite 
estimate $\hat\nu$ of $\nu$ is in some sense too small 
for this case where 
$\Sigma$ is exactly separable, 
$\hat\nu$ is generally large enough to make the shrinkage 
weight on $w$ nearly equal to one, which gives an estimate
that is nearly identical to the KMLE.

\subsection{Speech recognition}
Many data analysis tasks rely on accurate covariance estimates,
including tasks that are not specifically 
about covariance estimation. For example, 
quadratic discriminant analysis (QDA)
is a simple and popular method 
of classification that relies on estimates of the population 
means and covariances of each potential class to which 
new observations are to be assigned. Specifically, 
the score of 
a new observation with feature vector 
$y\in \mathbb R^p$ with respect to category $k\in \{1,\ldots, K\}$
is 
\[ 
  s_k(y) = (y-\hat \mu_k)^\top \hat \Sigma_k^{-1} (y-\hat \mu_k) + \ln |\hat \Sigma_k|, 
\]
where $(\hat \mu_k,\hat \Sigma_k)$ are estimates of the population 
mean and covariance of the feature vectors of objects in class $k$. 
If the frequencies of the different classes are equal, the classification 
rule is to assign the object with feature vector $y$ to the class 
with the minimum score. 
The accuracy of such a classification procedure will depend on, 
among other things, the accuracy of the mean and covariance 
estimates for each group. 
In cases where the feature vector $y$ is 
the vectorization of a matrix of features, 
we may consider using the core shrinkage estimator 
given by (\ref{eqn:ebest})
to make classifications, 
as an alternative to either
the unstructured MLE, 
the separable MLE, or other types of estimators. 

As a numerical illustration, we consider classification 
of spoken-word audio samples for 10 command words
 (``yes'', ``no'', ``up'', ``down'', ``left'', ``right'', ``on'', ``off'', 
``stop'',  
``go''), using the dataset provided by 
\citet{warden_2017} and described in 
\citet{warden_2018}. The data we consider include 
20,600 1-second long audio WAV files, with a per-word sample size 
ranging from 1,987 to 2,103 across the 10 words, representing 
between 989 and 1079 unique speakers for each word. 
We retain 100 audio samples per word for testing, and train our 
classifier on the remaining 19,600 audio samples. We do not make use of 
the fact that some speakers are represented multiple times in the dataset. 

A standard set of features for audio classification 
are mel-frequency cepstral coefficients (MFCCs), which 
describe an audio sample
in terms of a matrix whose dimensions represent 
periodicities in the power spectrum of the signal across 
time increments
\citep[Appendix A]{rao_manjunath_speech_2017}. 
For each audio sample in the 
dataset, we computed a $p_1\times p_2 = 99\times 13$ matrix 
of the first 13 mel cepstral coefficients across 99 time 
bins using the function {\tt melfcc} in the 
{\sf R}-package {\tt tuneR} \citep{ligges_et_al_2018}. 
Sample means and correlations for two of the words appear 
in Figures \ref{fig:mvup}  and \ref{fig:mvdown} 
(correlations instead of covariances are 
easier to visualize because of the large across-coefficient 
heteroscedasticity). 
The sample covariance matrices for these words 
are $p\times p = 1287 \times 1287$ matrices where, 
for example, the $99\times 99$ block in the upper left 
corner is the sample covariance matrix for the first 
cepstral coefficient across the 99 time points.

\begin{figure}
\centering{\includegraphics[height=2.75in]{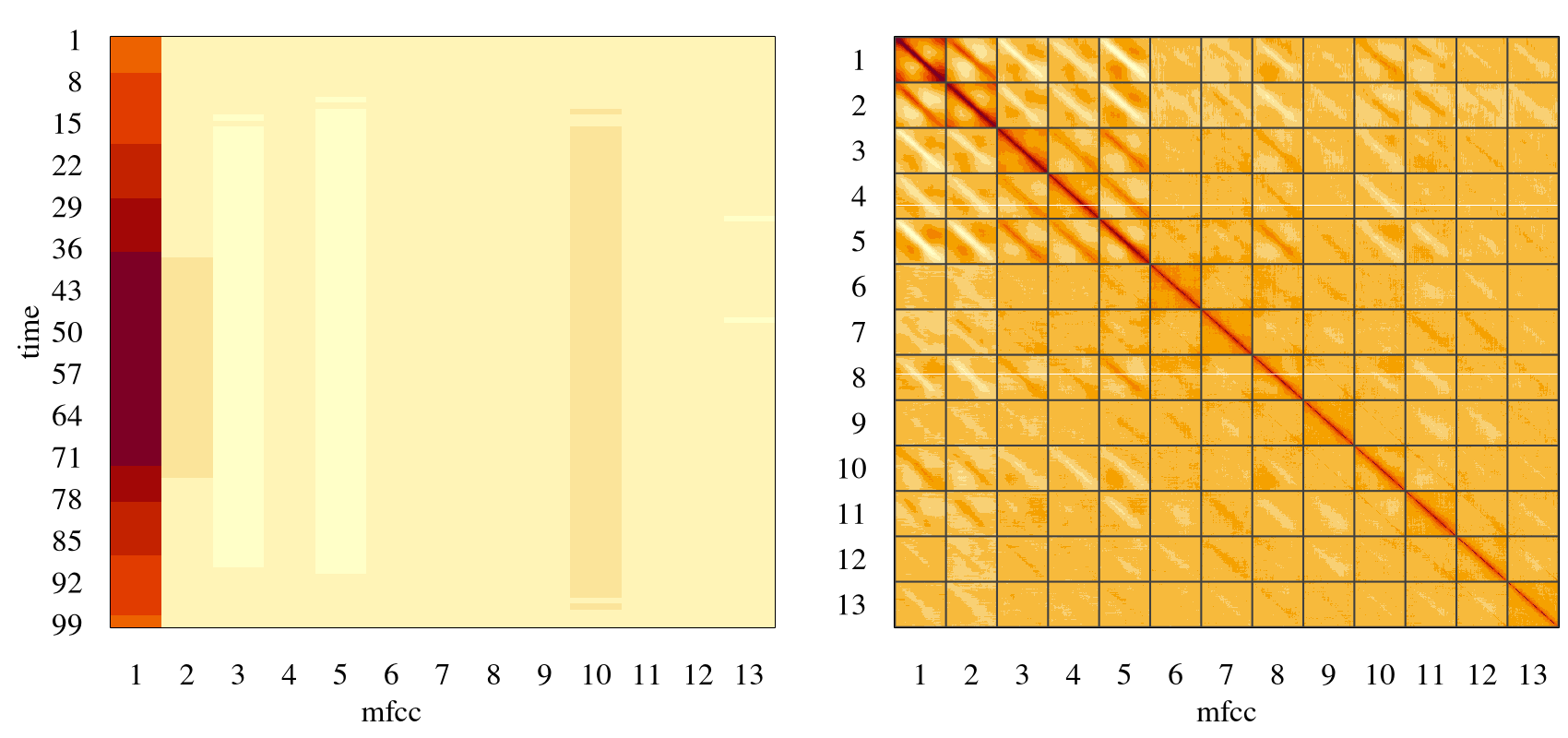}} 
\caption{Mean (left) and correlation (right) for MFCC's of the word ``up.''} 
\label{fig:mvup} 
\end{figure}

\begin{figure}
\centering{\includegraphics[height=2.75in]{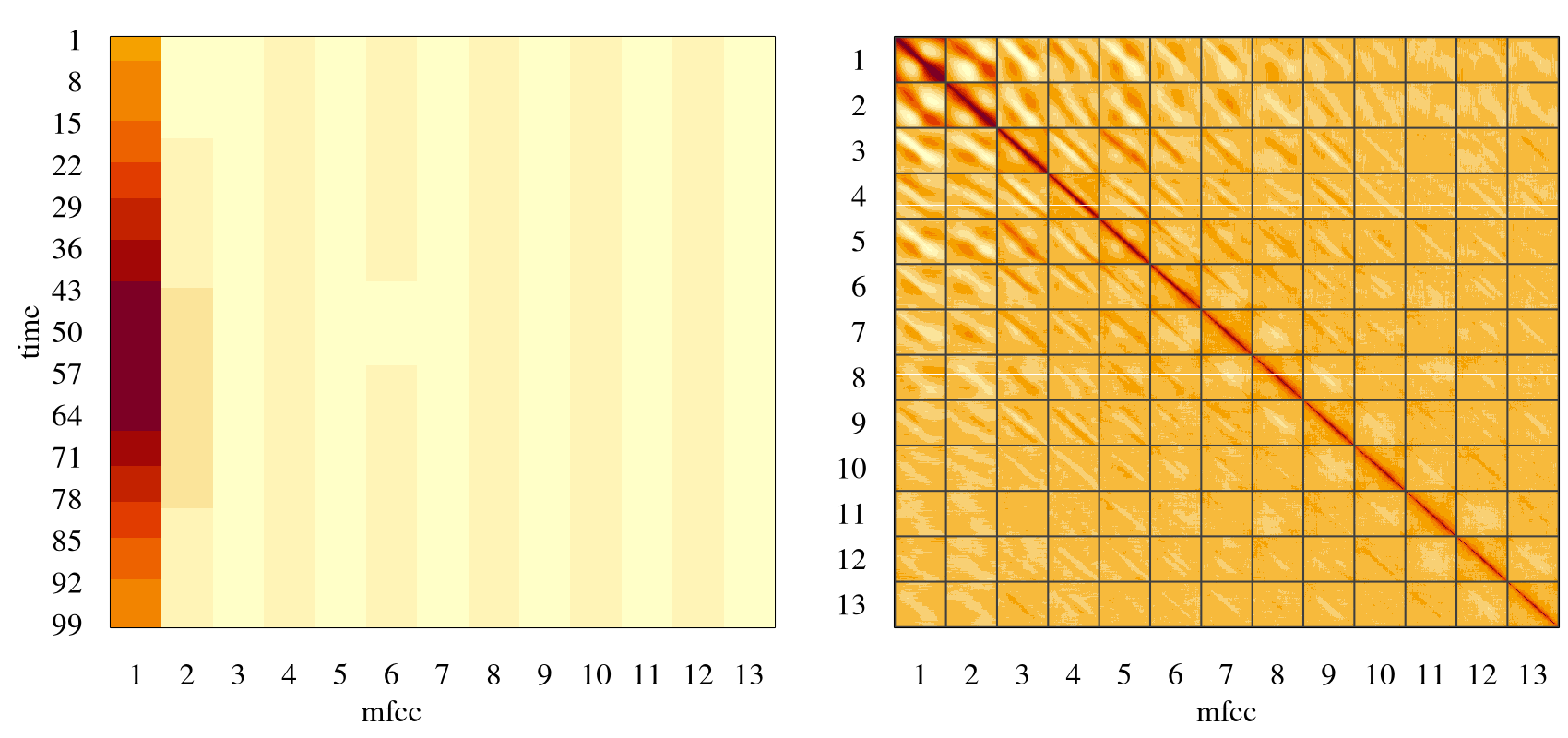}} 
\caption{Mean (left) and correlation (right) for MFCC's of the word ``down.''}
\label{fig:mvdown} 
\end{figure}

From the training data, 
we computed sample means and several different 
covariance estimates for each of the ten words. 
Our primary interest is in comparing prediction 
accuracy of the core shrinkage estimator to that of the 
unstructured and separable MLEs, 
but we also compute predictions 
using estimates that are partially pooled across groups. 
Quadratic discriminant analysis using partially pooled covariance 
estimates often have better performance than using 
class-specific sample covariance matrices, 
particularly when 
the sample size $n$ is not large compared to the dimension $p$. 
A variety of methods exist for choosing the pooling weights 
\citep{greene_rayens_1989,friedman_1989,rayens_greene_1991}. 
Here we use the approach outlined in 
\citet{greene_rayens_1989}, which is based on an inverse-Wishart 
hierarchical model 
for $\Sigma_1,\ldots, \Sigma_{10}$. The resulting 
partially pooled
covariance estimates (PPEs) are each roughly equal 
to a 32\%-68\% weighted average of the word-specific sample covariance 
and the pooled sample covariance matrices respectively.

\begin{figure}
\centering{\includegraphics[height=5.75in]{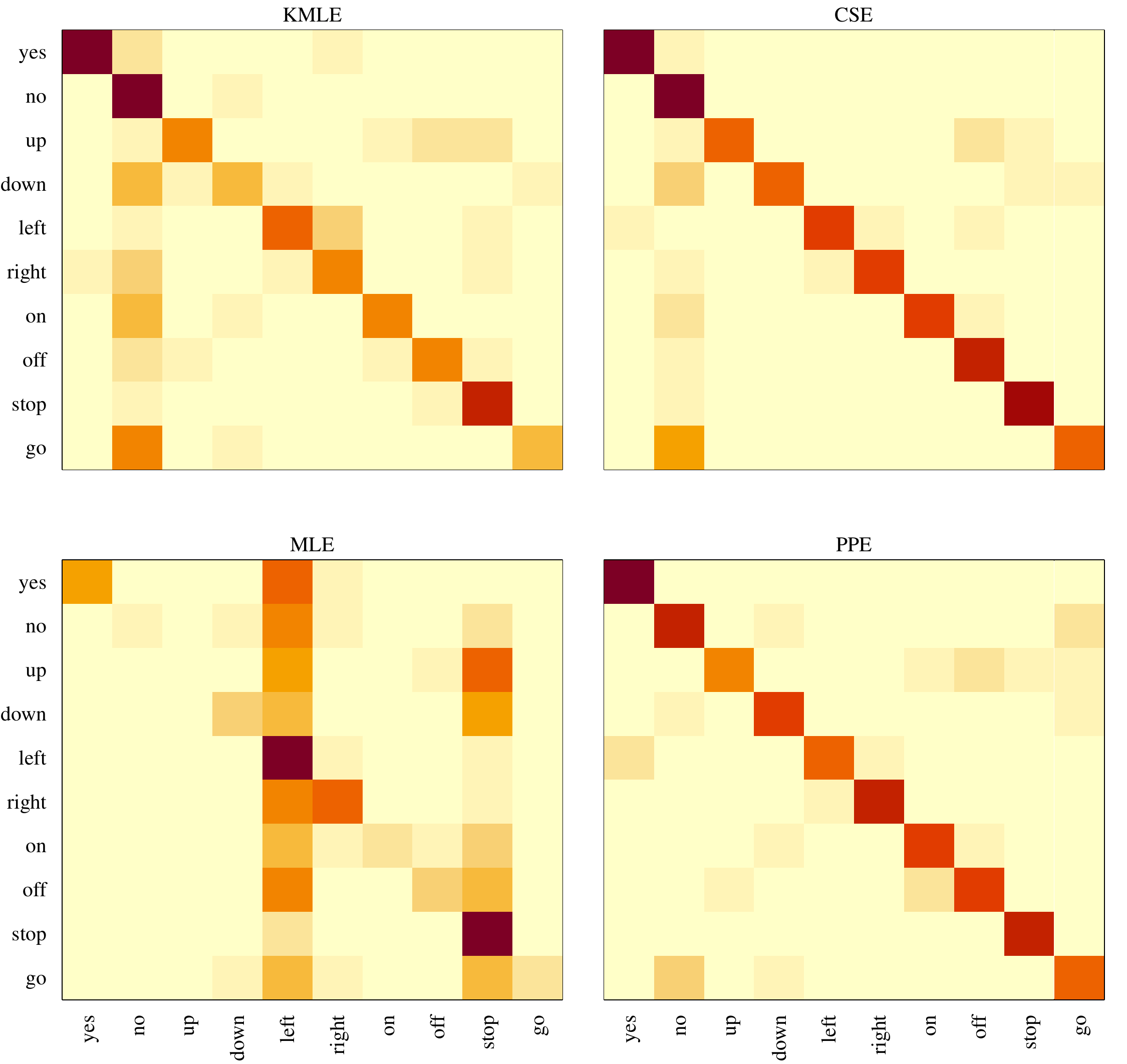}} 
\caption{Confusion matrices resulting from the four covariance estimates. Rows 
correspond to target words and columns correspond to predictions.}
\label{fig:confusion}
\end{figure}

Classifications for the 100 training observations 
were made using each of the covariance estimates. 
Confusion matrices 
are displayed in 
Figure \ref{fig:confusion}, with the true word classes
along the rows, and the predicted classes along the columns. 
For example, the word ``go'' is most frequently misclassified 
as ``no''.
From the figure,
QDA with the 
core shrinkage estimator appears to be substantially more accurate than 
using either the unstructured or separable MLEs, 
and is similar to using the partially pooled estimates. 
Rates of correct classification 
across all words  for all four QDA classifiers are given in 
Table \ref{tab:pcc}. 
The CSE performs better than the KMLE
for all words, and better than the unrestricted MLE 
for all words except ``left'' and ``stop''. However, this apparent 
good performance on these two words is misleading, as it is 
a result of this classifier assigning most words to being 
either in one of these two categories, as can be seen 
from Figure \ref{fig:confusion}. 
Additionally, the CSE is as good or better than 
the PPE for 
seven of the ten words. 
We note that the PPE 
is, like the CSE, a type of shrinkage estimator, 
although one that does not make use of the matrix structure of 
the data. 

\begin{table}
\begin{center}
\begin{tabular}{l|cccc}
  &  KMLE &  CSE & MLE & PPE  \\  \hline  
yes  &  0.69 & 0.79 & 0.37 & 0.82  \\
no  &  0.74 & 0.82 & 0.09 & 0.66  \\
up  &  0.38 & 0.51 & 0.04 & 0.46  \\
down  &  0.30 & 0.53 & 0.24 & 0.57  \\
left  &  0.44 & 0.60 & 0.77 & 0.51  \\
right  &  0.40 & 0.60 & 0.45 & 0.67  \\
on  &  0.41 & 0.59 & 0.16 & 0.58  \\
off  &  0.41 & 0.65 & 0.22 & 0.57  \\
stop  &  0.61 & 0.70 & 0.75 & 0.66  \\
go  &  0.30 & 0.50 & 0.14 & 0.48  
\end{tabular} 
\end{center} 
\caption{Rates of correct classification on the test dataset 
for the four  classifiers}
\label{tab:pcc}  
\end{table}

\section{Discussion}
Many classic estimators of covariance matrices are obtained 
by first computing the eigendecomposition of the sample covariance matrix
and then regularizing the resulting eigenvalues 
\citep{stein_1975,takemura_1983}. The core shrinkage estimator proposed in this article can be viewed analogously: the Kronecker-core decomposition of the sample covariance matrix is computed, and then the resulting core is regularized. 
However, while 
existing distributional results for the 
sample eigenvalues 
permit theoretical risk calculations for
unstructured covariance estimators, we lack such detailed knowledge 
of the distribution of sample core matrices. 
Further research on the distribution of sample cores 
could permit theoretical comparisons of different core shrinkage estimators. 

Empirical results from the speech recognition task in Section 4.1 
indicate that in this multi-group setting,
 two distinct types of shrinkage - towards separability and 
towards a common unstructured covariance matrix - both provided 
performance improvements.
This suggests that further improvements might be obtained 
with an estimator that combines these two types of shrinkage. Such 
an estimator could be obtained by empirically assessing the 
degree to which the covariance matrices are separable, as well 
as the degree to which they are similar to each other, and then shrinking 
the sample covariance matrices appropriately. 

The results in this article extend naturally 
to separable covariance models for tensor-valued data, that is, 
data arrays having 
three or more 
index sets. 
For example, an empirical Bayes 
covariance estimator 
that shrinks a sample covariance matrix towards 
a Kronecker product of several smaller covariance matrices, 
one for each index set, 
can be derived as in Section 3.2, using the 
same objective function (\ref{eqn:nuobj}) to determine the amount of shrinkage. 
A less straightforward extension would be an estimator
that adaptively shrinks towards an appropriate 
separable submodel, that is, submodels that are separable 
after various index sets of the data array have been collapsed.

\appendix 

\section*{Proofs} 
\begin{proof}[Proof of Proposition \ref{prop:krondivergence}]  
We first obtain an identity that relates the expectations 
in (\ref{eqn:kroneqn}) to the trace term in the divergence function $d(K:\Sigma)$. 
Letting $y$ be the vectorization of $Y$, 
for $(K_1,K_2) \in \mathcal S_{p_1}^+ \times \mathcal S_{p_2}^+$ we have
\begin{align*}  
\tr ( (K_2^{-1} \otimes K_1^{-1} ) \Sigma ) &= 
\Exp{ \tr ( (K_2^{-1} \otimes K_1^{-1} ) y y^\top ) }  \\
 &= \Exp{ y^\top (K_2^{-1} \otimes K_1^{-1} ) y } \\  
 &= \Exp{ \tr( Y^\top K_1^{-1} Y K_2^{-1} ) } \\
 &= \tr( K_1^{-1} \Exp{ Y K_2^{-1} Y^\top} ) = 
    \tr( K_2^{-1} \Exp{ Y^\top K_1^{-1} Y} ). 
\end{align*}
Therefore, for $K=K_2\otimes K_1$ the divergence function may be written 
\begin{align*} 
d(K:\Sigma)  
&= p_2 \ln|K_1|+ p_1 \ln |K_2| + \tr( K_1^{-1} \Exp{ Y K_2^{-1} Y^\top} ) \\
&= p_2 \ln|K_1|+ p_1 \ln |K_2| + \tr( K_2^{-1} \Exp{ Y^\top K_1^{-1} Y} ). 
\end{align*}
Now suppose that $\Sigma_2\otimes \Sigma_1\in \mathcal S_{p_1,p_2}^+$ 
minimizes the divergence. 
Then $\Sigma_1$ must also be the minimizer of the divergence
in $K_1$ when $K_2$ is fixed at $\Sigma_2$, that is, 
$\Sigma_1$ minimizes
$   p_2 \ln |K_1| + \tr( K_1^{-1} \Exp{Y \Sigma_2^{-1} Y^\top } )$
over $K_1\in \mathcal S_{p_1}^+$. It is well known 
\citep[Section 4.1]{anderson_2003}
that 
this function of $K_1$ is 
uniquely minimized by $\Exp{Y \Sigma_2^{-1} Y^\top }/p_2$, 
and so $\Sigma_1 = \Exp{Y \Sigma_2^{-1} Y^\top }/p_2$. 
Similarly, $\Sigma_2$ must equal 
$\Exp{Y^\top \Sigma_1^{-1} Y }/p_1$, and so $(\Sigma_1,\Sigma_2)$ 
is a solution to (\ref{eqn:kroneqn}). 

Conversely, 
let $f(K_1,K_2:\Sigma) = \ln|K_2\otimes K_1| + \tr(K_2 \otimes K_1)^{-1} \Sigma))$ be the divergence written as a real-valued function on
$\mathcal S_{p_1} \times \mathcal S_{p_2}$. 
Differentiating $f$ with respect to $(K_1,K_2)$ shows that 
the stationary points of $f$ are 
the solutions to (\ref{eqn:kroneqn}). 
Although $f$ is not convex, it is geodesically convex \citep{wiesel_2012}, and so by Corollary 3.1 of \citet{rapcsak_1991}, every stationary point of $f$ is a global minimizer of $f$. Thus if $(K_1,K_2)$ is a solution to (\ref{eqn:kroneqn}) then $K_2\otimes K_1$ is a minimizer of $d$. 
\end{proof}

\begin{proof}[Proof of Proposition \ref{prop:kronaction}]   
Let $A = A_2 \otimes A_1$. 
For each $K$, we have
\begin{align*} 
d(K: A \Sigma A^\top)  &= 
\ln |K| + \tr( K^{-1} A \Sigma A^\top )  \\
&=   \ln | A^{-1} K A^{-\top}| + \tr( (A^{-1}  K A^{-\top }) ^{-1} \Sigma)  + \ln | A A^\top |   \\
  & \equiv d( \tilde K : \Sigma) +\ln | A A^\top |, 
\end{align*} 
where $\tilde K = A^{-1} K A^{-\top}$.  
Note that for $A\in GL_{p_1,p_2}$, 
$\{ A^{-1} K A^{-\top} : K \in \mathcal S_{p_1,p_2}^+ \} = 
 \mathcal S_{p_1,p_2}^+$. 
By Proposition \ref{prop:krondivergence},  
$d(\tilde K: \Sigma)$ is minimized by $\tilde K = \Sigma_2\otimes \Sigma_1$, and so $d(K: A \Sigma A^\top)$  is minimized 
by $ K = A \tilde K A^\top =
 (A_2 \Sigma_2 A_2^\top) \otimes (A_1\Sigma_1 A_1^\top )$. 
\end{proof}

\begin{proof}[Proof of Corollary \ref{cor:kroncor}]   
Items 1 and 2 can be shown by noting that the unconstrained 
minimizer of $\ln |K| + \tr(K^{-1} \Sigma)$ over $K\in \mathcal S_p^+$ 
is $\Sigma$, and so if $\Sigma\in \mathcal S_{p_1,p_2}^+$ then 
the minimizer over $K\in \mathcal S_{p_1,p_2}^+$ is 
$\Sigma$ as well. Alternatively, 
item 1 can be shown by noting that $I_p = I_{p_2}\otimes I_{p_1}$, 
and confirming that $(I_{p_1}, I_{p_2})$ provide a solution to 
(\ref{eqn:kroneqn}) when $\Var{Y}= I_p$. Item 2 can also be shown this way, 
or
with 
Proposition \ref{prop:kronaction}: If $\Sigma = \Sigma_2\otimes \Sigma_1$ 
then 
\begin{align*}
k(\Sigma) & =  k( (\Sigma_2^{1/2}\otimes \Sigma_1^{1/2}  ) I_p 
(\Sigma_2^{1/2}\otimes \Sigma_1^{1/2}  ) ) \\
&=  (\Sigma_2^{1/2}\otimes \Sigma_1^{1/2}  )  k(I_p ) 
(\Sigma_2^{1/2}\otimes \Sigma_1^{1/2}  )  \\
& =  (\Sigma_2^{1/2}\otimes \Sigma_1^{1/2}  ) 
 (\Sigma_2^{1/2}\otimes \Sigma_1^{1/2}  )  = \Sigma. 
\end{align*}
Item 3 can also be obtained from 
Proposition \ref{prop:kronaction} by choosing (for example)
$A_1 = a I_{p_1} $ and $A_2=I_{p_2}$. 
Finally, if $\Var{Y} = \Sigma$ is diagonal then 
$\Exp{ y_i^\top A_1 y_{i'} } = 0$ for rows
$y_i$ and $y_{i'}$ of $Y$ for any matrix $A_1\in \mathbb R^{p_2\times p_2}$
unless $i=i'$. 
As a result, 
 $\Exp{ Y A_1 Y^\top }$ is diagonal, as is 
$\Exp{ Y^\top A_2 Y^\top}$ for the same reason. This implies 
that if 
$(\Sigma_1,\Sigma_2)$ is a solution to (\ref{eqn:kroneqn}) then both matrices are diagonal, as is their Kronecker product. 
\end{proof}

\begin{proof}[Proof of Proposition \ref{prop:coreid}]    
If $\Exp{YY^\top}/p_2=I_{p_1}$ and 
$\Exp{Y^\top Y}/p_1=I_{p_2}$ then $(I_{p_1},I_{p_2})$
is a solution to (\ref{eqn:kroneqn}) and so 
$k(C) = I_{p_2} \otimes I_{p_1} = I_p$. Conversely, 
if $k(C)=I_p$ then any solution to 
(\ref{eqn:kroneqn})  must be of the form 
$(cI_{p_1},c^{-1} I_{p_2} )$ for some $c>0$, which then 
implies that 
$\Exp{YY^\top}/p_2=I_{p_1}$ and 
$\Exp{Y^\top Y}/p_1=I_{p_2}$. Finally, let 
$y_j$ be the $j$th column vector of $Y$. 
Then 
\[ \Exp{ YY^\top }   = \sum_{j=1}^{p_2} \Exp{ y_j y_j^\top } =
   \sum_{j=1}^{p_2} \tilde C_{,j,,j}.  \]
\end{proof}

\begin{proof}[Proof of Proposition \ref{prop:coreprop}] 
Let $c(\Sigma)=C$, $k(\Sigma)=K$ and $h(K)=H$, 
so $\Sigma = H C H^\top$. 
By Proposition \ref{prop:kronaction}, $k( A \Sigma A^\top) = 
   A K A^\top = A H H^\top A^\top$.  
Let $\tilde K = A K A^\top$ and $\tilde H = h(\tilde K)$. 
Then 
$c(A\Sigma A^\top ) =   \tilde H^{-1} (A H)C ( AH)^\top \tilde H^{-\top}$.  
But by the definition of the square root function, we must have 
$\tilde H \tilde H^\top = \tilde K =  A H H^\top A^\top$, 
and so $\tilde H = A H R^\top$ for some $R\in \mathcal O_p$. Furthermore
this $R$ must be separable because both $\tilde H$ and $AH$ are separable. 
Thus $\tilde H^{-1} = R H^{-1} A^{-1}$ and item 1 of the result follows. 
If $A \in \mathcal H$ and $\mathcal H$ is a group, then $AH\in \mathcal H$, and so $\tilde H \equiv h( A H H^\top A^\top ) = A H$, 
giving item 2. 
\end{proof}

\begin{proof}[Proof of Proposition \ref{prop:reparam}]    
First we show that $f$ is a bijection. 
For any $\Sigma\in \mathcal S_{p}^+$, 
let $H = h(k(\Sigma))$ and $C=c(\Sigma)$. Then
\begin{align*}
g(f(\Sigma))  &=  
 H C H^\top \\
&=  H (H^{-1} \Sigma H^{-\top } ) H^\top = \Sigma. 
\end{align*} 
Conversely, let $(C,K)\in \mathcal C_{p_1,p_2}^+ \times
 \mathcal S_{p_1,p_2}^+$. Then with $H= h(K)$, we have 
\begin{align*} 
f(g(C,K))& =  f( H C H^\top ) \\
  &= (k(HCH^\top ), c( HCH^\top )). 
\end{align*}  
Since $H \in \mathcal S_{p_1,p_2}^+$, by 
Proposition \ref{prop:kronaction} we have
\begin{align*}
k( HCH^\top ) &= H k(C) H^\top  \\ 
   & =  H I H^\top  \\ 
  &= H H^\top = K. 
\end{align*} 
Finally, 
\begin{align*} 
c(HCH^\top) & = h( K )^{-1}  ( HCH^\top ) h( K )^{-\top}  \\
    & = H^{-1} (HCH^\top ) H^{-\top } = C, 
\end{align*} 
and so $f(g(C,K)) = (C,K)$.

We now show that the Kronecker covariance function 
$k$ is continuous, from which the continuity results for $f$ and $g$ follow. 
The space $\mathcal S_p^+$ is a complete Riemannian manifold with respect to 
the affine invariant metric $d_A: \mathcal S_p^+ \times \mathcal S_p^+ \rightarrow \mathbb R^+$ given by 
\[
 d_A( \Sigma, \tilde \Sigma ) =\Vert \log ( \Sigma^{-1/2} \tilde \Sigma 
  \Sigma^{-1/2})  \Vert, 
\]
where ``$\log$'' is the matrix logarithm 
\citep{bhatia_2007,higham_2008}.
Note that 
by the form of $d_A$ 
and the fact that $d_A(\Sigma,\tilde \Sigma) \leq 
  d_A(\Sigma,I_p) + d_A(I_p ,\tilde \Sigma)$, 
a subset of $\mathcal S_p^+$ is bounded under this metric 
if and only if the eigenvalues of its elements are bounded away
from zero and infinity. 

Let $\{ S_n\}$ be a sequence in 
$\mathcal S_{p}^+$ that converges to 
$\Sigma \in 
\mathcal S_{p^+}$ in this metric. 
Convergence of the sequence implies it is bounded, 
and so
there exists an interval 
$[a,b] \subset (0,\infty)$ that contains the eigenvalues of 
$S_n$ for all $n$. 
We now show that boundedness of $\{S_n\}$ implies that
the  sequence 
$\{K_n\}\equiv \{ k(S_n)\}$ is bounded. 
Recall that $K_n$ is the minimizer of the divergence 
$d$ over $\mathcal S_{p_1,p_2}^+$, and so 
 $d(K_n:S_n) \leq  d( I_p: S_n)$. 
Using this fact and the bounds on the eigenvalues of $\{S_n\}$, we have  
\[
  \sum_{j=1}^p ( \log l_{n,j} + a/l_{n,j})   =
d( K_n: a I_p) \leq  d( K_n: S_n )  \leq  d(I_p: S_n ) \leq p b, 
\]
where $l_{n,j}$ is the $j$th largest eigenvalue of $S_n$. Noting 
that $\log x + a/x$ is a convex function with a minimum at 
$x=a$, we have for each $k\in \{1,\ldots,p\}$ 
\[
\log l_{n,k} + a/l_{n,k}  \leq  
   pb - \sum_{j\neq k} ( \log l_{n,j} + a/l_{n,j} )  
  \leq  pb - (p-1)\times ( \log a + 1 ). 
\]
Since $\log x + a/x$ diverges as $x$ goes to zero or infinity, 
the above bound implies that there exists $[c,d] \subset (0,\infty)$
that contains 
$l_{n,k}$ for 
  all $n$ and $k$, that is, $\{K_n\}$ is bounded. 

Now let $\{ K_{n_s}\}$ be any convergent subsequence of 
$\{K_n\}$ and let $K = k(\Sigma)$. 
Let $K_{n_s} \rightarrow K^*$, and so 
$d(K_{n_s} : S_{n_s} )  \leq  d(K : S_{n_s})$. 
Since $d$ is jointly continuous in both of its arguments, taking 
the limit of the previous inequality gives 
$d(K^* : \Sigma )  \leq  d(K : \Sigma)$, which implies that 
$K^*=K$. 
This implies that $K_n\rightarrow K$ because
the closure of the bounded set $\{ K_n\}$ is itself bounded and therefore 
sequentially compact by the completeness of $\mathcal S_{p}^+$. Thus 
$k$ is continuous. Furthermore, since the topology of $\mathcal S_p^+$ under the affine invariant metric is the same as that under the Euclidean metric 
\citep[Theorem 2.55]{lee_2018}, $k$ is continuous for this metric space as well.
Finally, the functions  $f$ and $g$ are continuous because they are both compositions of the continuous function $k$ with other continuous functions. 
\end{proof}

\begin{proof}[Proof of Proposition \ref{prop:consistency}]    
We first find a limiting form for an objective
function from which $\hat \nu$ is obtained. 
To facilitate our analysis, we use the objective function 
$l_n(r,\hat C)  =  -2   \log L(n r)/n +p (\log\tfrac{n}{2}-1)$ with 
$L$ defined in (\ref{eqn:nuobj}), so that 
the estimated value of $\nu$ is 
$\hat \nu = n\times \arg \min_{r\geq (p+1)/n} l_n(r,\hat C) = 
    \arg\max_{\nu\geq p+1} L(\nu)$, where now we make explicit the 
dependence of the objective function on the sample core matrix $\hat C$.  
As a function of $(r,C) \in \mathbb R^+ \times \mathcal C_{p_1,p_2}^+$, 
the objective function is then $l_n(r,C) = a_n + b_n(r) + c_n(r,C)$ where
 $a_n = p (\log\tfrac{n}{2} -1)$ and 
\begin{align*}
b_n &= -\frac{2}{n}\log \left ( \frac{\Gamma_p(n(1+r)/2)}{\Gamma_p(nr/2)} \right ) +
 p(1+r) \log(1 + r + \delta) - pr\log(r + \delta)  \\
c_n &  =    (1+r)  \log \vert (1-w) C + w I_p\vert 
\end{align*}
where $w=(r+\delta)/(1+r+\delta)$ with $\delta = -(p+1)/n$. 
We will show that as $n\rightarrow \infty$, $a_n+b_n(r)$  
converges uniformly to zero for $r \in [\epsilon,\infty)$ and 
$c_n(r,C)$ converges uniformly to $l(r,C)$, where
\[
  l(r,C) =   (1+r) \log \vert C/(1+r)  + 
 r I_p/(1+r)   \vert .
\] 
We start by showing 
convergence of $c_n(r,C)$ to $l(r,C)$, i.e., that the  
difference between $w$ and $r/(1+r)$ is asymptotically negligible. 
To see this, 
recall that 
the log determinant of a matrix is a continuous function, and so 
is uniformly continuous on the compact set of 
convex combinations of core matrices and 
the identity.
Next, we have that $(1-w) C  + w I_p$ converges uniformly 
to $C/(1+r) + rI_p/(1+r)$, because the norm of their difference is 
$\Vert (w - \tfrac{r}{1+r})(I_p -C ) \Vert < \sqrt{p(p-1)} 
   |w-\tfrac{r}{r+1}|$, 
and $\vert w - \tfrac{r}{1+r} \vert = \tfrac{\delta}{(1 + r + \delta)(1+r)}$ converges to zero uniformly in $r$ for $r>0$. 

Next
 we use Stirling's approximation $\log\big(\Gamma(z)\big) = z\log(z) - z + \tfrac{1}{2}\log(2\pi/z) + O(z^{-1})$ on the 
multivariate gamma terms of $b_n(r)$. 
Letting $\delta_j = (1 - j)/n$, we have
\begin{align*} 
&  - \frac{2}{n}\bigg( \log\big(\Gamma_p\big(n(1+r)/2\big)\big) - \log\big(\Gamma_p(nr/2)\big)\bigg)    \\ 
& =    -\frac{2}{n}\sum_{j = 1}^p \log\big(\Gamma( \frac{n(1+r) + 1 - j}{2})\big) + \log\big(\Gamma(\frac{nr + 1 - j}{2}) \big )\\
 &    =   \sum_{j = 1}^p   \bigg( - (1 + r + \delta_j)\log\big(n(1 + r + \delta_j)/2\big) + (r + \delta_j)\log\big(n(r + \delta_j)/2\big) + 1 \bigg) 
    \\
  &  -  \frac{1}{n} \sum_{j=1}^p \log(\frac{r + \delta_j}{1  + r + \delta_j}) + O\big(n^{-1}(1+r)^{-1}\big) + O\big((nr)^{-1}\big). 
  \end{align*}
 The last three terms in the above expression converge uniformly to $0$ over $r \in [\epsilon,\infty)$ for any $\epsilon >0$. Adding $a_n$ and the 
remaining terms  of $b_n(r)$ gives $a_n+ b_n(r)$ being approximately 
equal to
  \begin{align*}
     - r\sum_{j = 1}^p\log \left ( \frac{(1+r+\delta_j)(r+\delta)}{(1 + r + \delta)(r + \delta_j)} \right ) - \sum_{j=1}^p\log\left ( \frac{1+r+\delta_j}{1 + r+\delta} \right )- \sum_{j = 1}^p \delta_j\log\left ( \frac{1+r+\delta_j}{r+\delta_j}\right ).
  \end{align*}
  The second and third sums above converge uniformly to zero over $r \in [\epsilon,\infty)$. Regarding the first sum, consider the ratio
  \begin{align*}
     \left (\frac{1+r+\delta_j}{1+r + \delta}\right )^r = \left (1 + \frac{\delta_j - \delta}{1 + r + \delta}\right )^{r + 1 + \delta}\left (1 + \frac{\delta_j - \delta}{1 + r + \delta}\right )^{-(1+\delta)}. 
  \end{align*}
The log of the second factor on the right converges to zero uniformly in $r$. For the first factor we have 
  \begin{align*}
      1 \leq \left( 1 + \frac{\delta_j - \delta}{1 + r + \delta}\right)^{1 + r + \delta} \leq e^{\vert \delta_j \vert + \vert  \delta \vert} \rightarrow 1 
  \end{align*} 
as $n\rightarrow \infty$, 
   where the first inequality follows from $\delta_j - \delta \geq 0$.
  Similarly,
  \begin{align*}
      1 \geq  \left (\frac{r+\delta}{r+\delta_j}\right )^r 
      = \left (1 + \frac{\delta - \delta_j}{r + \delta_j} \right )^r  
    \geq  \left (1 + \frac{\delta - \delta_j}{\epsilon + \delta_j} \right )^{\epsilon + \delta_j}  \left (1 + \frac{\delta - \delta_j}{r + \delta_j} \right )^{-\delta_j} \rightarrow  1
  \end{align*}
as $n\rightarrow \infty$.
Thus $a_n+b_n(r)$ converges uniformly to zero on $r\in [\epsilon,\infty)$ for 
any $\epsilon>0$.  

The above calculation shows that our objective function $l_n(r,C)$ 
converges uniformly to $l(r,C)$ for $(r,C) \in [\epsilon,\infty) \times \mathcal C_{p_1,p_2}^+$. We want to show that this limiting objective function is strictly increasing in $r$ if $C \neq I_p$, so in this scenario where $\Sigma$ is not separable 
the estimated weight on the sample Kronecker covariance 
converges to zero. 
To see that this is the case, 
let $c_1,\ldots c_p$ be the eigenvalues of $C$, so that
\begin{align*} 
l(r,C)  &= \sum_{j=1}^p (1+r) \log \left ( \frac{ r+c_j}{1+r}   \right )
  = \sum_{j=1}^p (1+r) \log \left ( 1+ \frac{c_j-1}{1+r}  \right ).
\end{align*}
The derivative of the $j$th term of the sum with respect to $r$ is 
  $\log ( 1+\tfrac{ c_j-1}{1+r} ) - \tfrac{ c_j-1}{ r+c_j}$. 
Since 
$\log(1+x) \geq x/(1+x)$ for $x > -1$ (with strict inequality for $x \neq 0$) this  derivative is positive for $(c_j-1)/(r+1)> -1$, or equivalently,
for $c_j> - r$, which holds for each $j=1,\ldots p$ because 
$C$ is positive definite. 
Additionally, because $C\neq I_p$ there is at least one $j$ for which 
$c_j\neq 1$, so at least one term 
in the sum  has a strictly positive derivative, making our 
objective function a strictly increasing function of $r$. 

Finally, let $\hat{C} = c(S)$ and $C = c(\Sigma)$.
We want to show that  $\hat{r}$, the minimizer of  $l_n(r,\hat{C})$ over $r\geq (p+1)/n$, converges in probability to zero
if $C\neq I_p$, 
or equivalently 
$\Pr(\hat{r} > \epsilon) \rightarrow 0$ for any $\epsilon>0$. 
By the result in the previous paragraph, $l(\epsilon/2,C) < l(\epsilon,C)$ 
and by the continuity of $l$ there is a ball $B$ around $C$ 
that does not contain $I_p$ such that
\begin{align*} 
  \inf_{\tilde C\in B} l(\epsilon,\tilde C) -  \sup_{\tilde C\in B}l(\epsilon/2,\tilde C) = \delta>0. 
\end{align*}
By the uniform convergence of $l_n$ to $l$, there is an $N$ such that 
 $|l_n(r,\tilde C)-l(r,\tilde C)|<\delta/2$ for $n>N$ and 
 all $\tilde C\in B$ 
and $r\geq \epsilon/2$. If $\hat C \in B$ then for any $r \geq \epsilon$ 
\begin{align*}
l_n(\epsilon/2,\hat C)  <  l(\epsilon/2,\hat C) + \delta/2  
 & \leq \sup_{\tilde C\in B} l(\epsilon/2,\tilde C ) + \delta/2  \\ 
 & = \inf_{\tilde C\in B} l(\epsilon,\tilde C) - \delta/2  \\
  & \leq l(\epsilon ,\hat C) - \delta/2 
  \\
  & \leq l(r,\hat C) - \delta/2
  \\
  & < l_n(r,\hat C)
\end{align*}
and so  $\hat r< \epsilon $ for $n>N$ and $\hat C \in B$. 
Thus $\Pr(\hat r>\epsilon ) \leq \Pr( \hat C \not\in B) \rightarrow 0$ 
as $n\rightarrow \infty$, because $B$ is a neighborhood of $C$ 
and $\hat C$ is consistent for $C$ by Corollary \ref{cor:khatconsistency}.  
Thus $ \hat r$ and $\hat w$ converge in probability to zero as $n\rightarrow 
\infty$ if $C\neq I$, that is, if $\Sigma$ is not separable. 
\end{proof}

\bibliography{kcd} 

\end{document}